\documentclass[reprint,longbibliography,aps,prb]{revtex4-2}
\usepackage[english]{babel}
\usepackage{amsmath,amsthm,amssymb,bbm,amsfonts,color,graphicx,xcolor,framed}
\usepackage{url}
\usepackage{tikz}
\usepackage{enumerate}
\usepackage{mathrsfs}
\usepackage{relsize}
\usepackage[colorlinks=true, allcolors=blue]{hyperref}
\usepackage[utf8]{inputenc}
\usepackage{listings}
\usepackage{stmaryrd}
\usepackage[braket, qm]{qcircuit}
\usepackage{graphicx}

\def\br{{\bf r}}

\def\bs{{\boldsymbol s}}
\def\bv{{\boldsymbol v}}

\newcommand{\om}{\omega}

\newcommand{\al}{\alpha}

\newcommand{\eq}[1]{Eq.~(\ref{eq:#1})}
\newcommand{\fig}[1]{Fig.~\ref{fig:#1}}
\renewcommand{\sec}[1]{Sec.~\ref{sec:#1}}

\newcommand{\app}[1]{Appendix~\ref{sec:#1}}

\newcommand{\thm}[1]{Theorem~\ref{thm:#1}}
\newcommand{\pro}[1]{Proposition~\ref{thm:#1}}
\newcommand{\lem}[1]{Lemma~\ref{thm:#1}}
\newcommand{\alg}[1]{Algorithm~\ref{alg:#1}}
\newcommand{\met}[1]{Method~\ref{met:#1}}
\newcommand{\defi}[1]{Definition~\ref{thm:#1}}

\renewcommand{\mod}{\mathrm{\,mod\,}}

\newcommand{\mc}[1]{\mathcal{#1}}

\let\perptmp\perp
\renewcommand{\perp}{{\! \mathsmaller{\perptmp}}}

\newcommand{\mrm}{\mathrm}

\def\br{{\bf r}}

\newtheorem{theorem}{Theorem}
\newtheorem{lemma}{Lemma}

\newtheorem{proposition}{Proposition}

\theoremstyle{definition}
\newtheorem{definition}{Definition}
\newtheorem{algorithm}{Algorithm}
\newtheorem{method}{Method}

\makeatletter
\def\@bibdataout@aps{%
  \immediate\write\@bibdataout{%
    @CONTROL{%
      apsrev41Control%
      \longbibliography@sw{%
        ,author="08",editor="1",pages="1",title="0",year="1"%
      }{%
        ,author="08",editor="1",pages="1",title="",year="1"%
      }%
    }%
  }%
  \if@filesw \immediate \write \@auxout {\string \citation {apsrev41Control}}\fi
}
\makeatother

\newcommand{\heavy}{\mathrm{heavy}}

\definecolor{blue}{rgb}{0.12156862745098039, 0.4666666666666667, 0.7058823529411765}
\definecolor{orange}{rgb}{1.0, 0.4980392156862745, 0.054901960784313725}
\definecolor{green}{rgb}{0.17254901960784313, 0.6274509803921569, 0.17254901960784313}
\definecolor{red}{rgb}{0.8392156862745098, 0.15294117647058825, 0.1568627450980392}
\definecolor{4}{rgb}{0.5803921568627451, 0.403921568627451, 0.7411764705882353}
\definecolor{5}{rgb}{0.5490196078431373, 0.33725490196078434, 0.29411764705882354}
\definecolor{6}{rgb}{0.8901960784313725, 0.4666666666666667, 0.7607843137254902}
\definecolor{7}{rgb}{0.4980392156862745, 0.4980392156862745, 0.4980392156862745}
\definecolor{8}{rgb}{0.7372549019607844, 0.7411764705882353, 0.13333333333333333}
\definecolor{9}{rgb}{0.09019607843137255, 0.7450980392156863, 0.8117647058823529}

\newcommand{\relem}[2]{
\vspace{.5em}
\noindent\textbf{\lem{#1}.}
\emph{#2}
}

\newcommand{\rethm}[2]{
\vspace{.5em}
\noindent\textbf{\thm{#1}.}
\emph{#2}
}

\newcommand{\realg}[2]{
\vspace{.5em}
\noindent\textbf{\alg{#1}.}
\emph{#2}
}

\begin{document}

\author{Joris Kattem\"olle}
\affiliation{Department of Physics, University of Konstanz, D-78457 Konstanz, Germany}
\title{Edge coloring lattice graphs}

\begin{abstract}
  We develop the theory of the edge coloring of infinite lattice graphs, proving a necessary and sufficient condition for a proper edge coloring of a patch of a lattice graph to induce a proper edge coloring of the entire lattice graph by translation. This condition forms the cornerstone of a method that finds nearly minimal or minimal edge colorings of infinite lattice graphs. In case a nearly minimal edge coloring is requested, the running time is $O(\mu^2 D^4)$, where $\mu$ is the number of edges in one cell (or `basis graph') of the lattice graph and $D$ is the maximum distance between two cells so that there is an edge from within one cell to the other. In case a minimal edge coloring is requested, we lack an upper bound on the running time, which we find need not pose a limitation in practice; we use the method to minimal edge color the meshes of all $k$-uniform tilings of the plane for $k\leq 6$, while utilizing modest computational resources. We find that all these lattice graphs are Vizing class~I. Relating edge colorings to quantum circuits, our work finds direct application by offering minimal-depth quantum circuits in the areas of quantum simulation, quantum optimization, and quantum state verification.
\end{abstract}
\maketitle

\section{Introduction}\label{sec:introduction}
An edge coloring of a graph $G$ is an assignment of colors to all its edges. An edge coloring is considered proper if no two edges of the same color are incident on the same vertex. It is called minimal if it additionally uses the least possible different colors. The number of colors used in a minimal edge coloring of $G$ is called the \emph{edge chromatic number} or chromatic index $\chi'(G)$. The problem of finding a minimal edge coloring of a graph is called the \emph{edge coloring problem}.  Vizing~\cite{vizing1964estimate,fiorini1977edge} proved that any finite simple graph $G$ with maximum degree $\Delta(G)$ belongs to one of two classes;
\begin{equation}
  \begin{array}{lll}
    \chi'(G)=\Delta(G)&\ \ &\text{(class I)},\\
    \chi'(G)=\Delta(G)+1&\ \ &\text{(class II)}.
  \end{array}\label{eq:vizing}
\end{equation}
Vizing's theorem was discovered independently by Gupta~\cite{gupta1967studies,stiebitz2012graph}, whose proof also applies to infinite graphs with a finite maximum degree. The problem of determining the class of a graph is NP-complete~\cite{holyer1981np_completeness}. Nevertheless, using the edge coloring algorithm by Misra and Gries \cite{misra1992constructive}, a proper edge coloring of a simple graph $G$ using at most $\Delta(G)+1$ colors is found in polynomial time.

We call a proper edge coloring of $G$ using at most $\Delta(G)$ colors a type-I coloring of $G$. If it exists, a type-I coloring is always minimal. We call a proper edge coloring of $G$ using at most $\Delta(G)+1$ colors a type-II coloring. A type-II coloring of a graph $G$ always exists. It is minimal in case $G$ is a class~II graph and \emph{nearly minimal} in case $G$ is a class~I graph, in which case the edge coloring uses just one excess color.

\begin{figure}[b]
  \centering
  \includegraphics[width=\columnwidth]{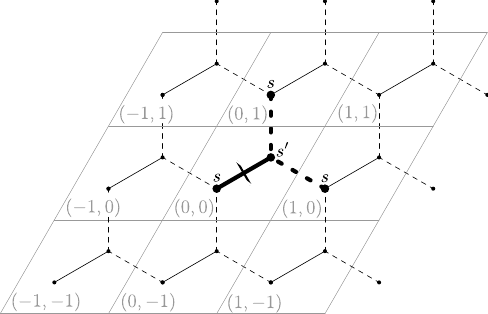}
  \caption{\label{fig:lattice_graph} Lattice graphs. A basis graph $B$ (bold vertices and edges), containing $\mu$ edges, induces a lattice graph $\mc G$ by the union of infinitely many translated copies of $B$. A patch $P_{n,m}$ of $\mc G$ is formed by the union of $n \times m$ copies of $B$. (In the figure, $n=m=3$, and $P_{n,m}$ is formed by all vertices and edges shown). Lattice graphs need not be planar, as symbolized by the slash through the edge between vertices $s$ and $s'$. The seeds $\mathscr S$ are the vertices inside the unit cell [parallelogram labeled $(0,0)$]. An edge-colored patch $C(P_{n,m})$ induces an edge coloring $\mc C$ of $\mc G$ by the union of infinitely many translated copies of $C(P_{n,m })$ (not illustrated here).
  }
\end{figure}

In this paper, we present a method that finds type-I or type-II edge colorings of lattice graphs. These graphs are illustrated in \fig{lattice_graph} and introduced formally in \sec{lattice_graphs}. However efficient, no standard methods for the type-I or type-II edge coloring of lattice graphs can be used because lattice graphs are infinite. To find a type-II edge coloring of a lattice graph, our method (\sec{edge_coloring_lattice_graphs}) searches a patch of the lattice graph that is self-loop-free and allows a type-II coloring after imposing periodic boundary conditions on the patch. We prove that such a patch is always found. Subsequently, the patch is type-II edge colored and the periodic boundary conditions are lifted. We prove that the resulting colored patch can be repeated indefinitely to form a type-II (and hence also proper) edge coloring of the entire lattice graph. The self-loop-free property of the wrapped patch plays a crucial role in this proof.

Similarly, to obtain a type-I edge coloring, a patch is searched that is self-loop-free and permits a \emph{type-I} coloring after imposing periodic boundary conditions. We prove that, if such a patch is found, it can be repeated indefinitely to produce a type-I (and therefore always \emph{minimal}) edge coloring of the lattice graph. However, it is no longer guaranteed that such a patch is found. First, the lattice graph may be class~II. Secondly, even if we assume that the lattice graph is class~I, it remains an open question whether there exists a patch that is self-loop-free and permits a type-I coloring after imposing periodic boundary conditions on that patch.

We implement and apply our edge coloring method to a plethora of lattice graphs (\sec{explicit_minimal_edge_colorings}), each time requesting a type-I edge coloring of the lattice graph. For all 1317 lattice graphs, a type-I coloring was found without prior knowledge of the class of these graphs. This demonstrates that, despite the lack of theoretical guarantees, in practice the method is effective in finding type-I and hence minimal edge colorings of lattice graphs. Among the lattice graphs used are the meshes of all $k$-uniform tilings of the plane for $k\leq 6$, which are hence all class~I.  We therefore hypothesize that the $k$-uniform tilings of the plane are class~I for all $k\geq 1$. At the same time, we construct a (planar) lattice graph that provably class~II and also use our method to edge color it minimally.

For practical applications, one is typically not interested in a type-I or type-II edge coloring of an infinite lattice graph, but rather in a type-I or type-II edge coloring of a finite but possibly large patch of that lattice graph. In addition to providing type-I or type-II edge colorings of infinite lattice graphs, our method provides significant computational advantages for the type-I or type-II edge coloring of these finite but large patches. This is because our method typically needs to find a type-I or type-II edge coloring of a small subgraph of the large patch. Once this is achieved, the solution can be extended efficiently to patches of arbitrary size. In contrast, methods that do not leverage the translational symmetry of lattice graphs have to find the type-I or type-II edge coloring of the entire large patch.

One important application of minimal edge coloring lattice graphs is in quantum computing, where proper edge colorings can be used to represent quantum circuits for tasks in quantum simulation, quantum optimization, and quantum state verification (\sec{previous_work_and_applications}). In this context, patches of minimally edge-colored lattice graphs correspond to minimal-depth quantum circuits.

\section{Lattice graphs}\label{sec:lattice_graphs}
For an initial understanding of the edge coloring method (\sec{edge_coloring_lattice_graphs}), the definitions in \fig{lattice_graph} may suffice. Readers mainly interested in the practical applicability of the method may proceed directly to \fig{coloring} for an initial understanding of the method and then to \sec{explicit_minimal_edge_colorings} for its application. However, the formal results in \sec{analysis} necessitate formal definitions, which we provide here. The terminology is derived from terminology standard in solid state-physics~\cite{ashcroft1976solid}.

A two-dimensional \emph{Bravais lattice} $\tilde{\mc{V}}$ [including points labeled $s$ (not $s'$) in \fig{lattice_graph}] is an infinite set of points in the Euclidean plane obtained by acting on a single point with a symmetry group generated by translations along the vectors $\bv_1,\bv_2\in\mathbb R^2$; $\tilde{\mc V}=\{x\,\bv_1+y\,\bv_2 \mid (x,y)\in \mathbb Z^2\}$. In this paper, we assume 2-dimensional lattices for ease of notation throughout. However, all methods and results can be straightforwardly generalized to any dimension by de- or increasing the dimension of the vectors.

A \emph{geometric lattice with seeds} $\mc V$ is a Bravais lattice in which a set $\mathscr S$ (points inside the parallelogram $(0,0)$ in \fig{lattice_graph}), containing vectors $\br\in \mathbb R^2$, known as seeds, is translated instead of a single point~\footnote{In Ref.~\cite{ashcroft1976solid}, this is called a `lattice with a basis'. However, to avoid confusion, we reserve the use of `basis' for the basis graph and use `seeds' instead of `basis' here. This is the terminology from Refs.~\cite{sanchez2019acquiring,sanchez2021integer,medeiros2018synthesizing}};
\begin{equation}\label{eq:seed_lattice}
\mc V=\{\br+x \bv_1+y \bv_2 \mid \br \in \mathscr S, (x,y) \in \mathbb Z^2\}.
\end{equation}
Without loss of generality, we assume that all seeds are in the unit cell, which is the region delimited by the parallelogram defined by the vectors $\bv_1,\bv_2$. A Bravais lattice is a geometric lattice with a single seed.

In this paper, a \emph{graph} (without qualification) is a tuple $G=(V(G),E(G))$, with vertices $V(G)$ and edges $E(G)$ countable sets that may contain self-loops and multi-edges. A \emph{self-loop} is an edge that connects a vertex with itself and a \emph{multi-edge} is a set of edges between a single pair of vertices. (For distinguishability of the edges in a multi-edge, each edge in a multi-edge is labeled with a label unique to that multi-edge.) Hence, a \emph{self-loop-free graph} may contain multi-edges but not self-loops. A \emph{simple graph} does not contain self-loops nor multi-edges. Whenever we want to emphasize that a graph may contain self-loops or multi-edges, we refer to it as a \emph{multigraph}.

We define a \emph{geometric lattice graph} $\mc G$ as a graph obtained by the union of translated copies of a \emph{geometric basis graph} $B$, similar to how a geometric lattice with seeds is generated by repeating the set of seeds $\mathscr S$. The vertices $V(B)$ of the geometric basis graph are vectors $\br\in\mathbb R^2$ and the edges $E(B)$ are hence sets of the form $\{\br,\bs\}$, which contain two vertices $\br,\bs$. In \fig{lattice_graph}, the basis graph of the honeycomb lattice graph is shown as the Y-shaped figure with bold vertices and lines. Formally, we define the geometric lattice graphs as infinite graphs of the form $\mc G=(V(\mc G), E(\mc G))$, where $V(\mc G)=\mc V$ as in \eq{seed_lattice} and with edges
\begin{equation}
  \begin{aligned}
    E(\mc G)&=\mathlarger{\{}\{\br+x \bv_1+y \bv_2,\bs+x \bv_1+y \bv_2\}  \\ &\phantom{=\{}\mid \{\br,\bs\}\in E(B),\, (x,y) \in \mathbb Z^2\mathlarger{\}}.
  \end{aligned}
\end{equation}

The specific Cartesian coordinates of the vertices of geometric lattice graphs are irrelevant for the minimal edge coloring problem. Moreover, using Cartesian coordinates may lead to floating-point issues when identifying equivalent vertices from two translated copies of the geometric basis graph. Furthermore, we may want to edge color graphs with no inherent geometrical meaning. Therefore, we primarily use (nongeometric) basis graphs, which are induced by (nongeometric) lattice graphs. They are essentially geometric lattice graphs where the vertices are not general vectors in the Euclidean plane but tuples of the form $(dx,dy,s)$, with $(dx,dy)$ denoting the cell in which a vertex is located (also see the gray cell labels in \fig{lattice_graph}), and $s$ being some unique identifier to distinguish seeds in a cell. For simplicity, we assume the seeds are identified by consecutive integers $s\in S=\{0,\ldots,n_S-1\}$.

\begin{definition}[Basis graph]\label{thm:basis_graph}
  A \emph{basis graph} $B=(V(B),E(B))$ with seed numbers $S=\{0,\ldots,n_S-1\}$ is a finite, simple graph, with all vertices of the form
\begin{equation}
  (dx,dy,s)\in V(B).
\end{equation}
 Here, $(dx,dy)\in\mathbb Z^2$ labels the \emph{cell} the vertex is in, and $s\in S$. Any vertex of $B$ of the form $(0,0,s)$, is called a \emph{seed} of $B$ and all other vertices of $B$ are called nonseeds.  An edge of $B$ is \emph{redundant} if it can be removed without altering the lattice graph induced by $B$. We demand that $B$ does not contain edges from nonseeds to nonseeds and that $B$ does not contain redundant edges.
\end{definition}

\begin{definition}[Lattice graph]\label{thm:lattice_graph}
A \emph{lattice graph} $\mc G=(V(\mc G),E(\mc G))$ induced by a basis graph $B$ is a graph obtained by the union of infinitely many translated copies of $B$;
\begin{equation}
  \begin{aligned}\label{eq:lattice_graph}
V(\mc G) &= \{v+(x,y,0)\mid v\in V(B),\,(x,y)\in \mathbb Z^2\},\\
  E(\mc G) &= \mathlarger{\{}\{v+(x,y,0),v'+(x,y,0)\} \\ &\phantom{=\{}\mid \{v,v'\} \in E(B),\,(x,y)\in \mathbb Z^2\mathlarger{\}}.
  \end{aligned}
\end{equation}
\end{definition}
By this definition, lattice graphs are simple, but they need not be planar.

The demand that $(i)$ $B$ does not contain redundant edges is without loss of generality, since it does not change the possible lattice graphs that can be constructed. Also, the demand that $(ii)$ $B$ does not contain edges from nonseeds to nonseeds is without loss of generality. Consider a graph $B'$ that is a basis graph, with the only exception that it contains edges from nonseeds to  nonseeds. Each of these edges can be translated [by adding $(x,y,0)$ for some $(x,y\in \mathbb Z^2)$ to both vertices] so that one of the vertices becomes a seed. The resulting basis graph $B$ will induce the same lattice graph as $B'$.

As an example, consider the honeycomb lattice graph (\fig{lattice_graph}). It is induced by the basis graph $B$ with vertices $V(B)=\{a,b,c,d\}$, where $a=(0,0,0)$, $b=(0,0,1)$, $c=(1,0,0)$, and $d=(0,1,0)$, and edges $E(B)=\mathlarger{\{}\{a,b\},\{b,c\},\{b,d\}\mathlarger{\}}$. The addition of the edge $\{a,(0,-1,1)\}$ to $E(B)$ would render either $\{a,(0,-1,1)\}$ or $\{a,c\}$ redundant $(i)$. The basis graph $B'$, where, for example, the edge $\{a,b\}$ is translated so that it entirely lies inside the cell $(1,0)$, generates the same honeycomb lattice graph as $B$ $(ii)$.

\begin{definition}[Patch]
  An $n$ by $m$ \emph{patch} $P_{n,m}=P_{n,m}(B)$ of a lattice graph $\mc G$ induced by a basis graph $B$ is a graph constructed by the union of $n$ by $m$ translated copies of $B$,
\begin{equation}
  \begin{aligned}\label{eq:patch}
V(P_{n,m})&= \{v+(x,y,0)\mid v\in V(B),\,(x,y)\in \mathbb Z_n \times \mathbb Z_m)\},\\
  E(P_{n,m}) &= \mathlarger{\{}\{v+(x,y,0),v'+(x,y,0)\} \\ &\phantom{=\{}\mid \{v,v'\} \in E(B),\,(x,y)\in\mathbb Z_n \times \mathbb Z_m \mathlarger{\}},
  \end{aligned}
\end{equation}
where $\mathbb Z_i=\{0,1,\ldots,i-1\}$.
\end{definition}
It follows from the definition of basis graphs that a patch of a lattice graph is simple and does not contain redundant edges. Any basis graph can be seen as a patch for which $n=m=1$. In contrast to the basis graphs, patches with $n\neq 1$ or $m\neq 1$ contain edges from nonseeds to nonseeds.

A graph $C(G)$ is an edge coloring of a graph $G$ if it is equal to $G$ except for an additional assignment of `colors' $c\in \mathbb N_0$ to all edges. With $G=P_{n,m}$, $C[P_{n,m}]$ is an edge colored patch of $n$ by $m$ basis graphs,
\begin{equation}
  \begin{aligned}
    V[C(P_{n,m})]&=V(P_{n,m})\\
    E[C(P_{n,m})]&=\mathlarger{\{(\{v,v'\},c_{v,v'})\mid \{v,v'\} \in E(P_{n,m})\mathlarger\}}. \label{eq:edge_coloring}
  \end{aligned}
\end{equation}
This graph induces an edge colored lattice graph $\mc C=(V(\mc C),E(\mc C))$ by
\begin{equation}
  \begin{aligned}
  V(\mc C) &= \{v+(n\,x,m\,y,0)\mid v\in V[C(P_{n,m})],\,(x,y)\in \mathbb Z^2\},\\
  E(\mc C) &= \mathlarger{\{}(\{v+(n\,x,m\,y,0),v'+(n\,x,m\,y,0)\},c)\\ &\phantom{=\{}\mid (\{v,v'\},c_{v,v'}) \in E[C(P_{n,m})],\,(x,y)\in \mathbb Z^2\mathlarger{\}}.
  \end{aligned}\label{eq:edge_colored_lattice_graph}
\end{equation}
The absence of redundant edges in $B$ assures $\mc C$ is an edge coloring of $\mc G$. If $C(P_{n,m})$ (at given $n,m$) induces an edge coloring $\mc C$ of a lattice graph $\mc G$, we call $C(P_{n,m})$ a \emph{coloring basis graph} of $\mc G$.

\section{Method}\label{sec:edge_coloring_lattice_graphs}

\begin{figure}
  \centering
  \includegraphics[width=\columnwidth]{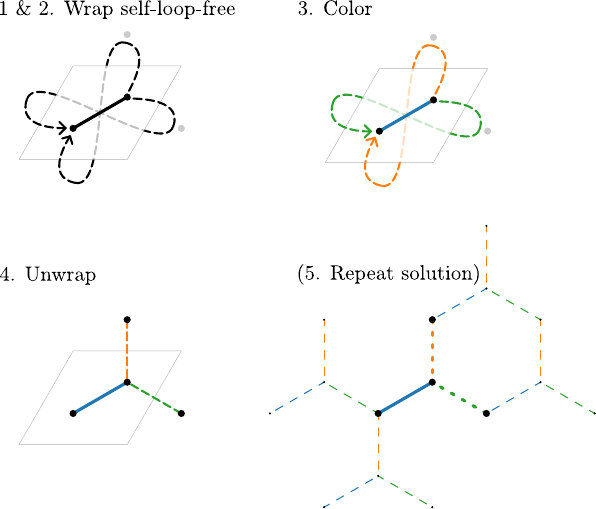}
  \caption{\label{fig:coloring} Main stages of \met{edge_coloring}. The preprocessing stage is to define the lattice graph by its basis graph (\fig{lattice_graph}). Then, stages 1--3 are repeated with increasing patch size until a patch is found at stage 3 that is self-loop-free and permits a type-$t$ ($t=1,2,3$) edge coloring. At stage 4, this patch is unwrapped, retaining its edge coloring. As a postprocessing stage (stage 5), the infinite lattice graph is generated purely abstractly, or a finite patch is generated explicitly.
  }
\end{figure}

\label{sec:method}
A preprocessing stage is to define the (possibly geometric) lattice graph $\mc G$ by a basis graph $B$ as defined in \defi{basis_graph} and \fig{lattice_graph}. The method receives $B$, a trit $t\in\{1,2,3\}$, and an initial patch size $(n,m)$ as input, where $t$ indicates whether the objective is to find a type-I ($t=1$), type-II ($t=2$) or proper ($t=3$) edge coloring. Below, we interleave the method (italic text) with further definitions and explanations (regular text). The main stages of the method are depicted in \fig{coloring}. Throughout this paper, we initially choose $n=m=1$, so that initially $P_{n,m}=B$, unless specified otherwise.

\begin{method}[Edge color lattice graphs]\label{met:edge_coloring}
  \

\emph{1. Wrap.---Construct the patch $P_{n,m}$ and wrap it to obtain the multigraph $\tilde P_{n,m}$.}

Periodic boundary conditions are imposed on a vertex $(x,y,s)\in V(P_{n,m})$ by
\begin{equation}
  \begin{aligned}
  w[(x,y,s)]:=(x \mod n, y \mod m,s).
\end{aligned}
\end{equation}
Here, we use the `floored divisor' definition of the $\mod$ operator; $x\mod n=x-n\lfloor x/n\rfloor$. The function $w$ naturally induces a map from simple graphs $G$ (such as $P_{n,m}$) to multigraphs by applying $w$ to all vertices and to all vertices in all edges of $G$, not discarding any self-loops and not merging any equal edges that may arise. We denote this induced map by $w$ as well.

The map $w$ is not an invertible map from simple graphs to multigraphs because the modulo operation is not invertible. We therefore construct an invertible map $W$ we call \emph{wrapping}. It is initially defined on edges $\{v_1,v_2\}$ (not vertices), by

 \begin{equation}\label{eq:wrap}
W(\{v,v'\}):=(\{w(v),w(v')\},\{v,v'\}).
 \end{equation}
 The map $W$ naturally induces a map $W$ from edge sets of simple graphs [such as $E(P_{n,m})$] to edge sets of multigraphs by element-wise action; $W[E(P_{n,m})]=\{W(e)\mid e\in E(P_{n,m})\}$. Since $P_{n,m}$ is simple, the edge label carrying the old edge data, $\{v,v'\}$, serves as a unique label distinguishing the edges in a multi-edge. The map $W$ can once more be extended to a map $W$ from simple graphs to multigraphs by acting nontrivially on the edge set of the graph only; $W(P_{n,m})=(V(P_{n,m}),W[E(P_{n,m})])$. This defines the \emph{wrapped patch} $\tilde P_{n,m}$,
 \begin{equation}
\tilde P_{n,m}:=W(P_{n,m}).
\end{equation}
In \fig{coloring}, the arrowheads and gray vertices make it visually clear how the wrapping should be undone, and therefore represent the old edge data.

\emph{2. No self-loops.---Assert that $\tilde P_{n,m}$ is self-loop-free. If not, choose a new patch size $(n,m)$ and return to stage 1.}

In \sec{analysis} and \app{proofs} (\thm{proper_edge_coloring}), we prove that the absence of self-loops in $\tilde P_{n,m}$ is necessary and sufficient for the output of the current method to induce a \emph{proper} edge coloring of the lattice graph $\mc G$. After the exposition of the method, we also illustrate the necessity of the self-loop-free property by example.

If $\tilde P_{n,m}$ has self-loops, a new candidate patch needs to be proposed. We choose  a new $(n,m)$ such that the size $n\times m$ of the patch is nondecreasing. That is, we choose $(n,m)$ from the sequence $ a=((1,1),(1,2),(2,1),(1,3),(3,1),(2,2),\ldots)$. Such a sequence gives a bijection between pairs $a_i$ and the nonzero integers $i$, so that, if the sequence is freely traversed, each patch size is eventually encountered.

\emph{3. Color.--- Type-$t$ color the wrapped patch $\tilde P_{n,m}$ to obtain the edge colored and wrapped patch $\tilde C$. If a type-$t$ coloring of $\tilde P_{n,m}$ is not found, choose a new patch size $(n,m)$ and return to stage~1.}

To type-$t$ color $\tilde P_{n,m}$ means to properly edge color $\tilde P_{n,m}$ using either any number of colors ($t=3$), $\Delta (\tilde P_{n,m})+1$ colors ($t=2$), or $\Delta (\tilde P_{n,m})$ colors ($t=1$). In \sec{analysis} and \app{proofs} (\lem{max_degree}), we prove that $\Delta (\tilde P_{n,m})=\Delta(\mc G)$ if $\tilde P_{n,m}$ is self-loop-free, which is guaranteed by stage~2 of the current method.

Assuming that a type-$t$ coloring $\tilde C = \tilde C(\tilde P_{n,m})$ is found, we may write its vertices and edges as [cf. \eq{edge_coloring}]
\begin{equation}
  \begin{aligned}
    V[\tilde C(\tilde P_{n,m})]&=V(\tilde P_{n,m}),\\
    E[\tilde C(\tilde P_{n,m})]&=\{(\{w(v),w(v')\},\{v,v'\},c_{v,v'}) \\
    \ \ &\mid (\{w(v),w(v')\},\{v,v'\})\in E(\tilde P_{n,m})\}.
  \end{aligned}
\end{equation}
The crucial point for the subsequent stages of the method is that, in the above, we may write $c_{v,v'}$ because an edge $(\{w(v),w(v')\},\{v,v'\})\in\tilde E(P_{n,m})$ is uniquely identified by $\{v,v'\}$.

Any method may be used to seek $\tilde C$. In our implementation, we use the satisfiability modulo theories (SMT) solver Z3~\cite{moura2008z3,githubz3}. The solver is guaranteed to find a type-$t$ coloring if such a coloring exists and reports that no such coloring exists otherwise. In case a proper edge coloring is sought $(t=3)$, for every pair of edges incident on the same vertex the constraint is added to the SMT formula that states that the colors of those edges cannot be equal. If a type-II coloring is sought ($t=2$), an additional constraint is that the total number of colors used is $\Delta (\tilde P_{n,m})+1$. If a type-I coloring is sought for ($t=1$), instead the additional constraint is that the total number of colors used is $\Delta (\tilde P_{n,m})$.

A proper edge coloring exists for any $\tilde P_{n,m}$. By Vizing's theorem, a type-II coloring is guaranteed to exist if $\tilde P_{n,m}$ is simple. In \sec{analysis} and \app{proofs} (\lem{simple}), we prove that a simple $\tilde P_{n,m}$ is always encountered as the sequence of patch sizes $a$ is traversed. Nonetheless, if $\tilde P_{n,m}$  contains a multi-edge, a type-II coloring may still exist, in which case it will be found by the SMT solver. It is an open question whether a wrapped patch that permits a type-I coloring will be encountered as the sequence of patch sizes $a$ is traversed, even if it is assumed that $\mc G$ is class-I.

\emph{4. Unwrap.---Unwrap $\tilde C$ to obtain the edge-colored patch $C$.}

Unwrapping an edge reinstates the old edge, carrying with it the color that was added in stage~2,
\begin{equation}
 U(\{w(v),w(v')\},\{v,v'\},c_{v,v'})=(\{v,v'\},c_{v,v'})
\end{equation}
It naturally induces a map on graphs of the form of $\tilde C$ by acting on all edges of $\tilde C$. If applied to $\tilde C$, $U$ produces an edge coloring $C$ of $P_{n,m}$. Thus, we have
\begin{equation}
 U(\tilde C)=C.
\end{equation}
To summarize the output thus far,  $C=U(\tilde C\{W[P_{n,m}(B)]\})$ at given $(n,m)$. Note that $C$ is an edge coloring of $P_{n,m}$. Return $C$, $n$ and $m$.  \hfill \qedsymbol
\end{method}

In \sec{analysis} and \app{proofs} (\thm{correctness}), we show that $C$ induces a type-$t$ coloring of the lattice graph $\mc G$. Akin to \eq{patch}, as a postprocessing stage, edge colored `super' patches of $N$ by $M$ edge coloring basis graphs $C$ can be obtained with \eq{edge_colored_lattice_graph} after substituting $\mathbb Z^2\to \mathbb Z_N\times \mathbb Z_M$. Such a construction is used to construct edge colored super patches in \sec{explicit_minimal_edge_colorings}.

We now motivate stage~2 (`no self-loops') of the method by example, postponing the general and careful analysis to \sec{analysis} and \app{proofs}. Suppose the input of the method is formed by $B,t$, with $t=3$ and $B$ a basis graph of the square lattice graph that contains a single seed. This basis graph has vertices $V(B)=\{\mrm a,\mrm b,\mrm c\}$, with $\mrm a=(0,0,0)$, $\mrm b=(1,0,0)$, and $\mrm c=(0,1,0)$, and edges $E(B)=\mathlarger{\{}\{\mrm a,\mrm b\},\{\mrm a,\mrm c\}\mathlarger{\}}$. Then, stage~1 (`wrap') produces the wrapped basis graph $\tilde B$. This graph has two edges, both of which are self-loops of the seed $\mrm a$. Now suppose stage~2 (`no self-loops') is skipped, so that $\tilde B$ is passed to stage~3 (`color') despite containing a self-loop. $\tilde B$ can only be properly edge colored using two colors. Stage 4 (`unwrap') now produces a graph $C$ with $E(C)=\{(\{\mrm a,\mrm b\},c_{\mrm a,\mrm b}),(\{\mrm a,\mrm c\},c_{\mrm a,\mrm c})\}$ for some $c_{\mrm a,\mrm b}\neq c_{\mrm a,\mrm c}$. Because the square lattice graph $\mc G$ has maximum degree four, $C$ cannot induce a proper edge coloring of $\mc G$. When, instead, stage~2 (`no self-loops') is included, the input to that stage, $\tilde P_{n,m}$, is rejected until we reach a patch with $(n,m)=(2,2)$. It is then properly edge colored and unwrapped. The resulting edge colored patch induces a proper edge coloring of the lattice $\mc G$ (\fig{archimedean}).

\section{Explicit minimal edge colorings}\label{sec:explicit_minimal_edge_colorings}
In this section, we introduce various collections of lattice graphs and apply \met{edge_coloring} to find a type-I and hence minimal edge coloring for each. Images of all 1318 edge colored lattice graphs are available at \cite{Note2}. \met{edge_coloring}, including pre- and postprocessing, is implemented in \textsc{Python} and is also available at \footnote{See Supplemental Material at [URL will be inserted by publisher] for all code and data related to this paper. These are also available at at \url{https://github.com/kattemolle/ecolpy}}. The computationally most demanding task of searching type-I colorings of candidate wrapped patches is offloaded to Z3. Results and wall-clock times are obtained by running the code on a 2019 MacBook Pro (16 inch, 2,6 GHz 6-Core Intel Core i7, 16 GB RAM) using a single thread.

\subsection{Archimedean and Laves tilings}\label{sec:archimedean_and_laves}
A tiling of the plane is a covering of the plane by closed sets (tiles) without gaps or overlaps of nonzero area~\cite{grunbaum2016tilings}. The vertices and edges of the tiles, or the `mesh' of a tiling, naturally define an infinite planar graph. A tiling is termed periodic if it exhibits translational symmetry in at least two nonparallel directions~\cite{grunbaum2016tilings}. Thus, a periodic tiling naturally defines a lattice graph.

The Archimedean tilings are tilings of the plane by convex regular polygons with the property that all vertex figures have the same form. That is, all vertices, including the surrounding geometrical edges, are congruent~\cite{grunbaum2016tilings} figures. There are exactly 11 Archimedean lattices, as was first discovered by Kepler in his Harmonice Mundi from 1619~\cite{kepler1619harmonices, aiton1997harmony}. The Archimedean tilings are periodic and hence \met{edge_coloring} can be used to color their edges. Minimal edge colorings of all Archimedean tilings were found in 0.3\,s and are shown in \fig{archimedean}.

The Laves tilings use a single shape of tile (monohedral), where, for each vertex separately, all angles between two circularly adjacent edges are equal (vertex regular)~\cite{grunbaum2016tilings}. These tilings are also the face-dual tilings of the Archimedean tilings. Consequently, there are 11 Laves tilings. The square lattice is self-dual, and the honeycomb and triangular tilings are mutually self-dual, leaving  8 Laves tilings that are not Archimedean. Minimal edge colorings of these 8 tilings were found in 0.4\,s and are shown in \fig{laves}.

\newcommand{\lf}[1]{\begin{minipage}{0.25\textwidth}\noindent \includegraphics[scale=.15]{#1.pdf}\end{minipage}}
\newcommand{\nm}[1]{\begin{minipage}{0.25\textwidth}\noindent #1 \end{minipage}}

\newcommand{\lff}[1]{\begin{minipage}{0.25\textwidth}\noindent \includegraphics[scale=.15]{#1.pdf}\end{minipage}}
\newcommand{\nmm}[1]{\begin{minipage}{0.25\textwidth}\noindent #1 \end{minipage}}

\begin{figure*}
  \begin{tabular}{llll} \lf{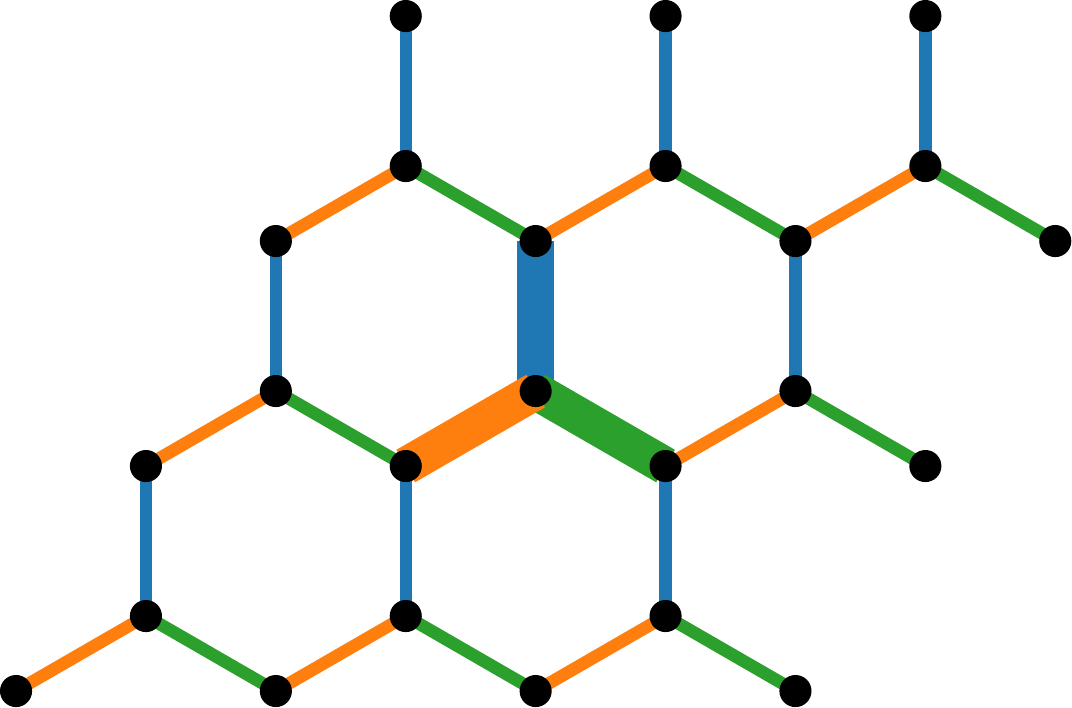}&\lf{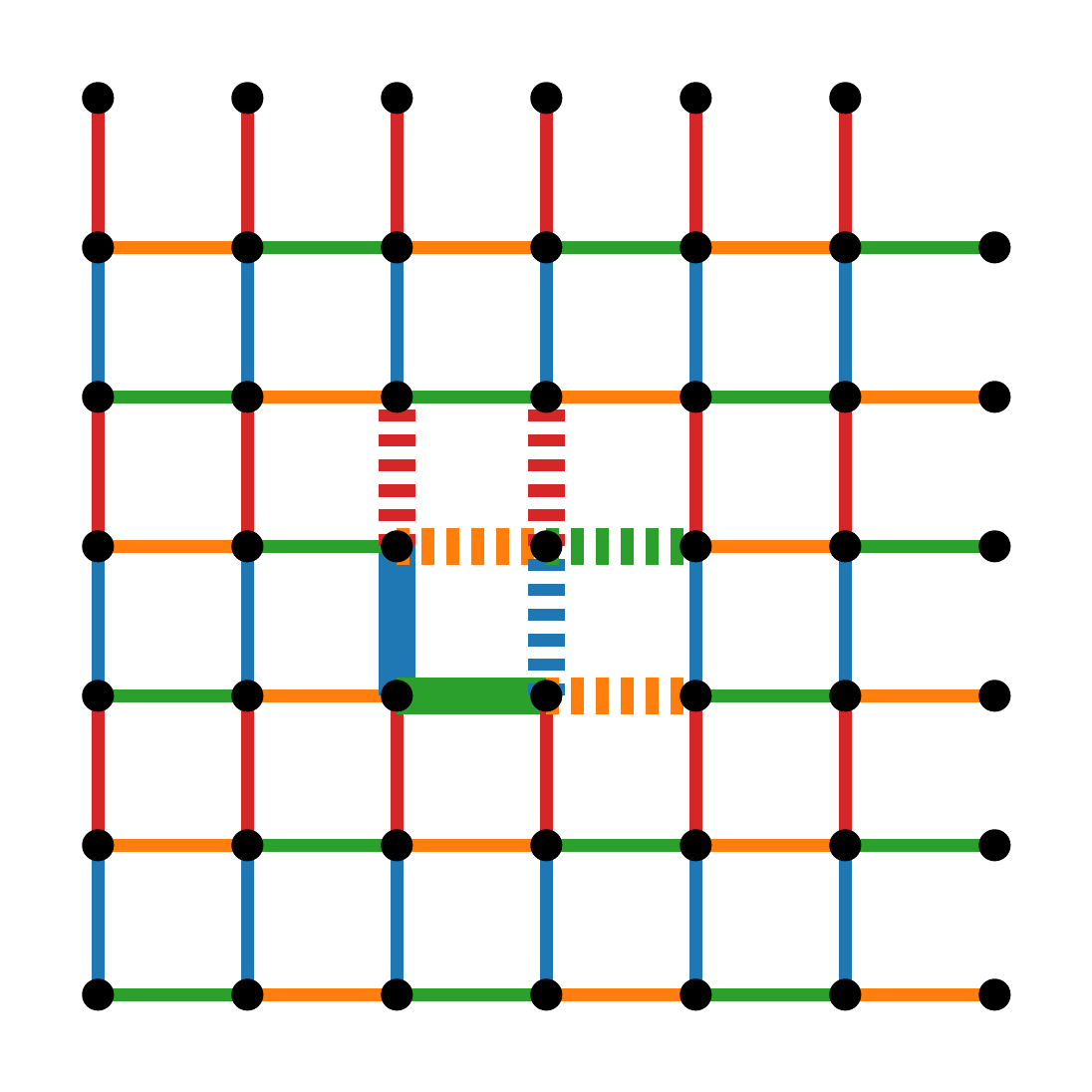}&\lf{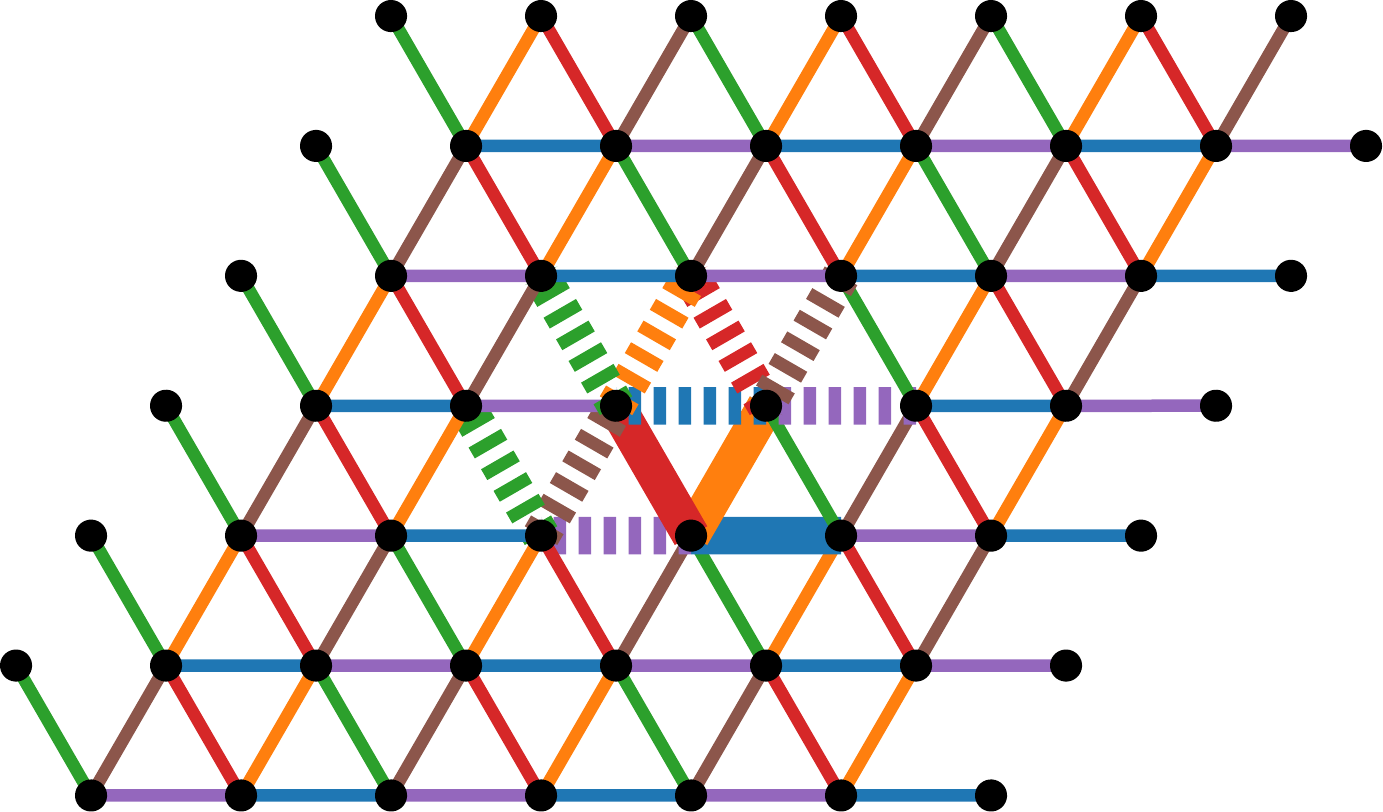}&\lf{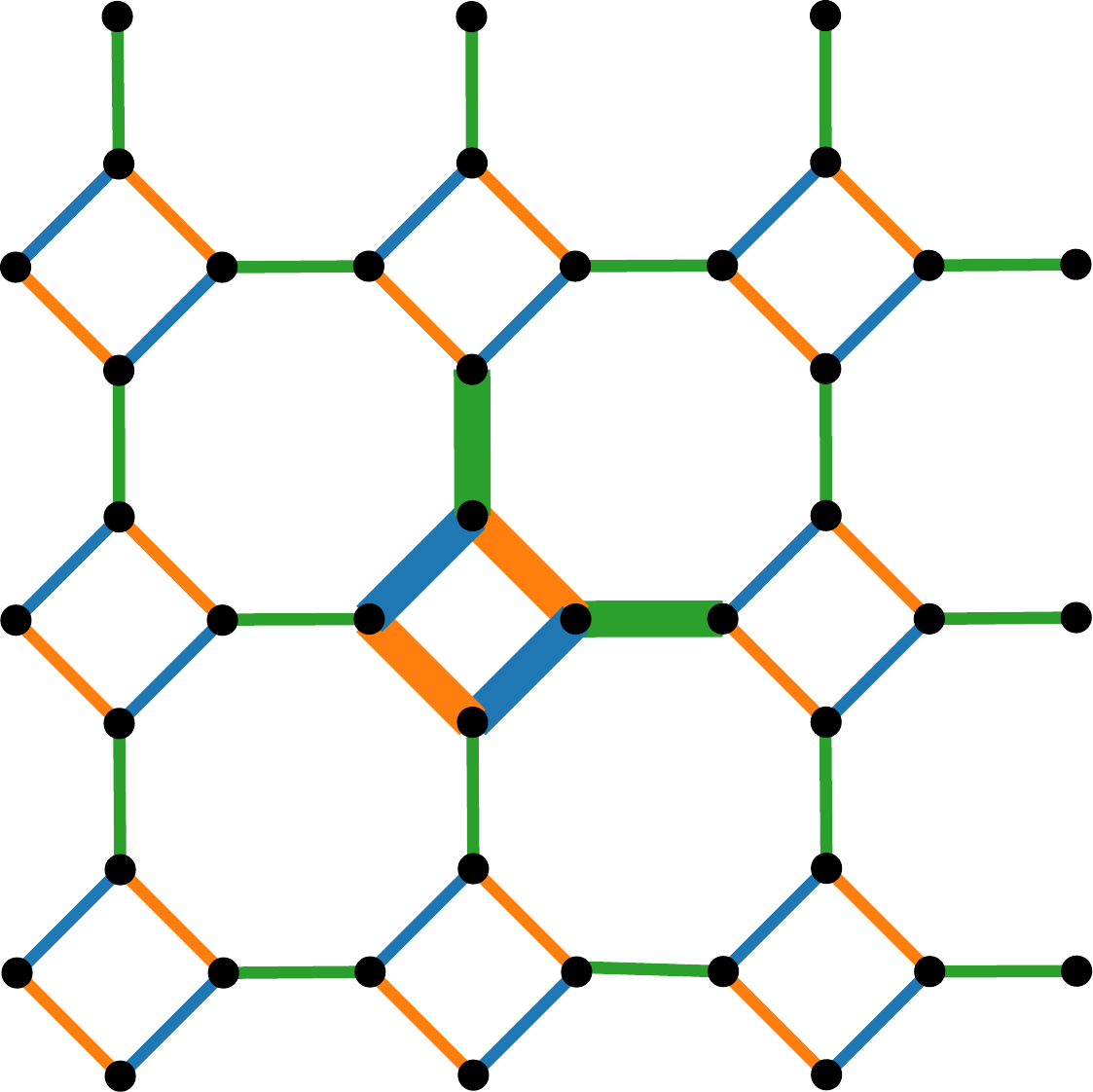}
    \\
    \nm{$(6^3)$\\honeycomb\\hexagonal\\hextille}&
    \nm{$(4^4)$\\square\\grid\\quadrille}&
    \nm{$(3^6)$\\triangular\\deltille} & \nm{$(4,8^2)$\\square-octagon\\truncated-square\\bathroom-tile\\truncated-quadrille}
  \end{tabular}
  \begin{tabular}{cccc}
    \lf{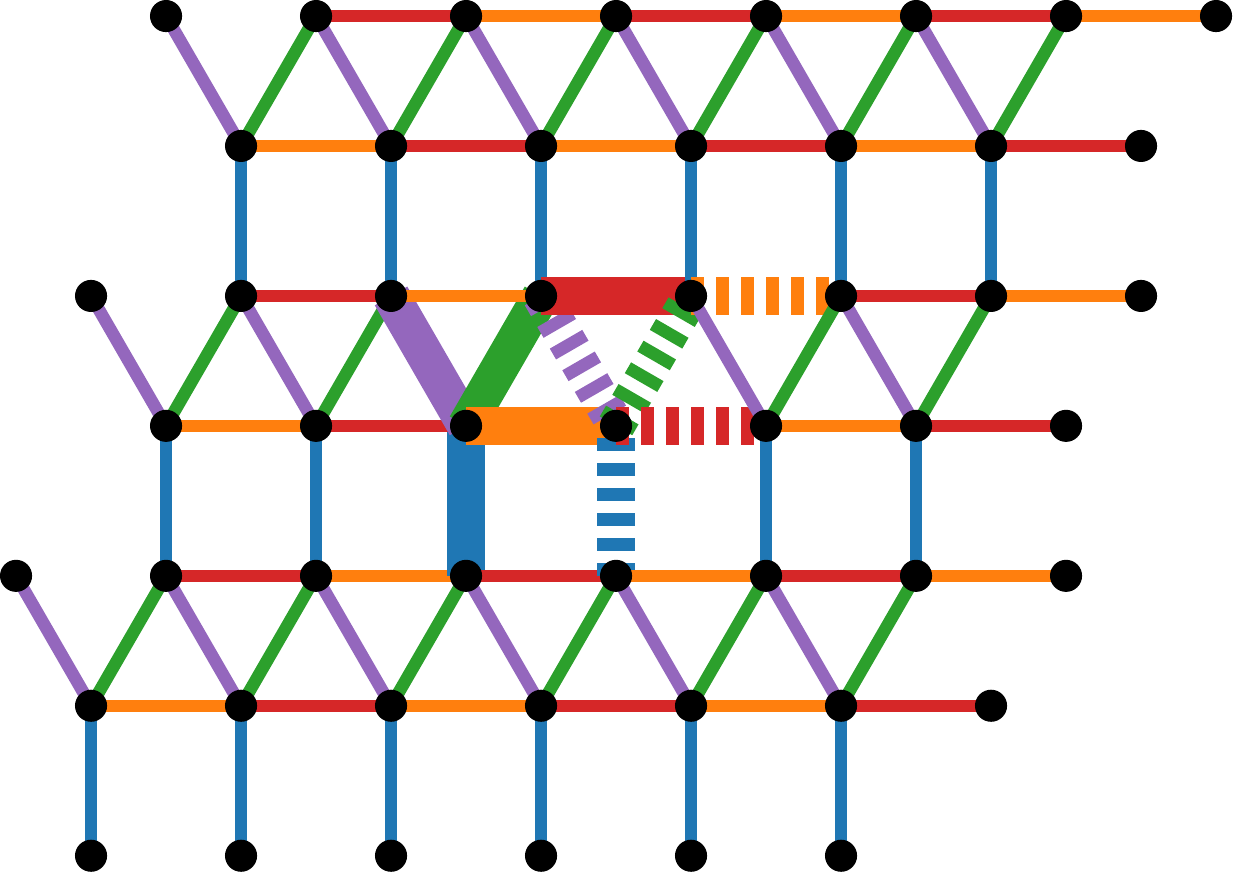}&\lf{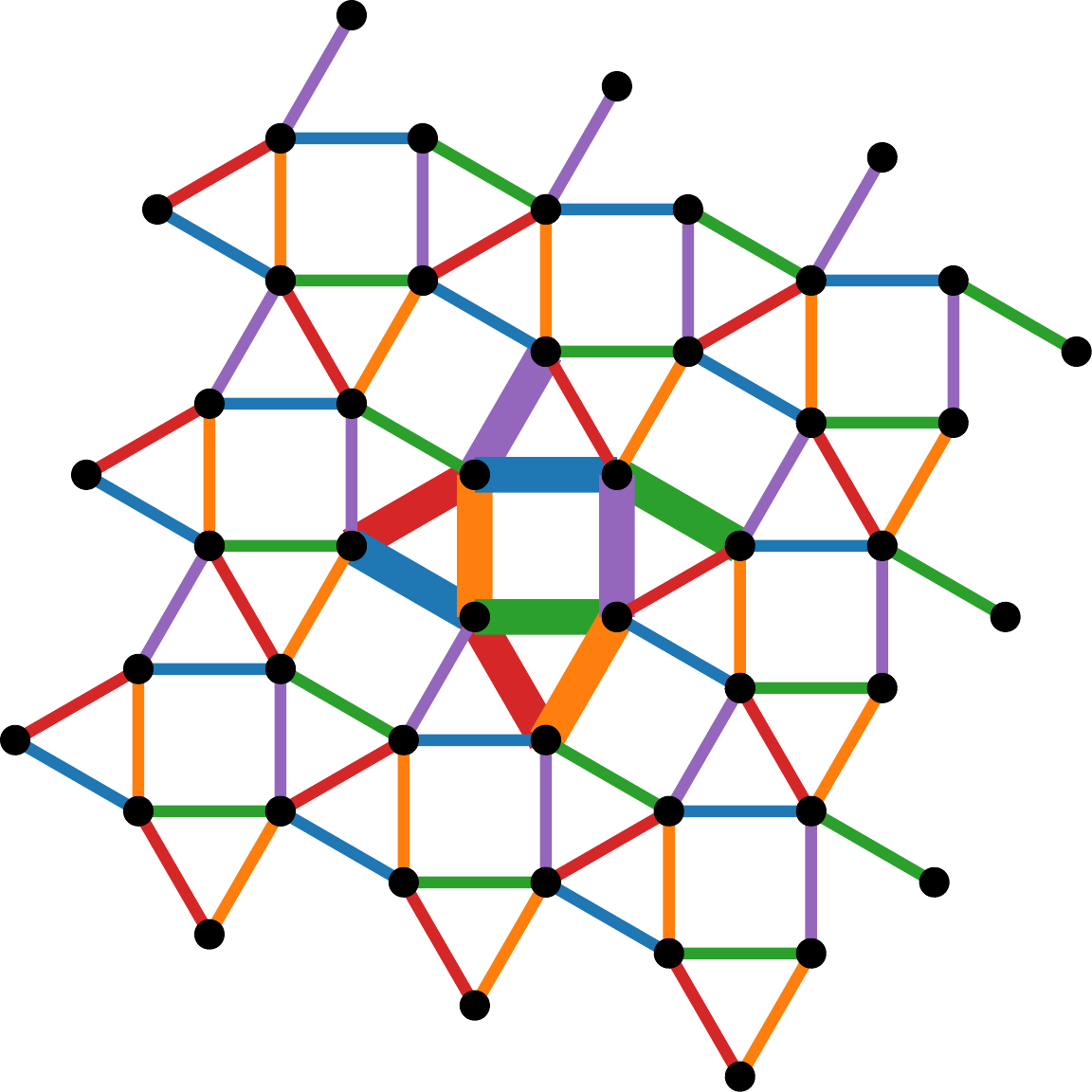}&\lf{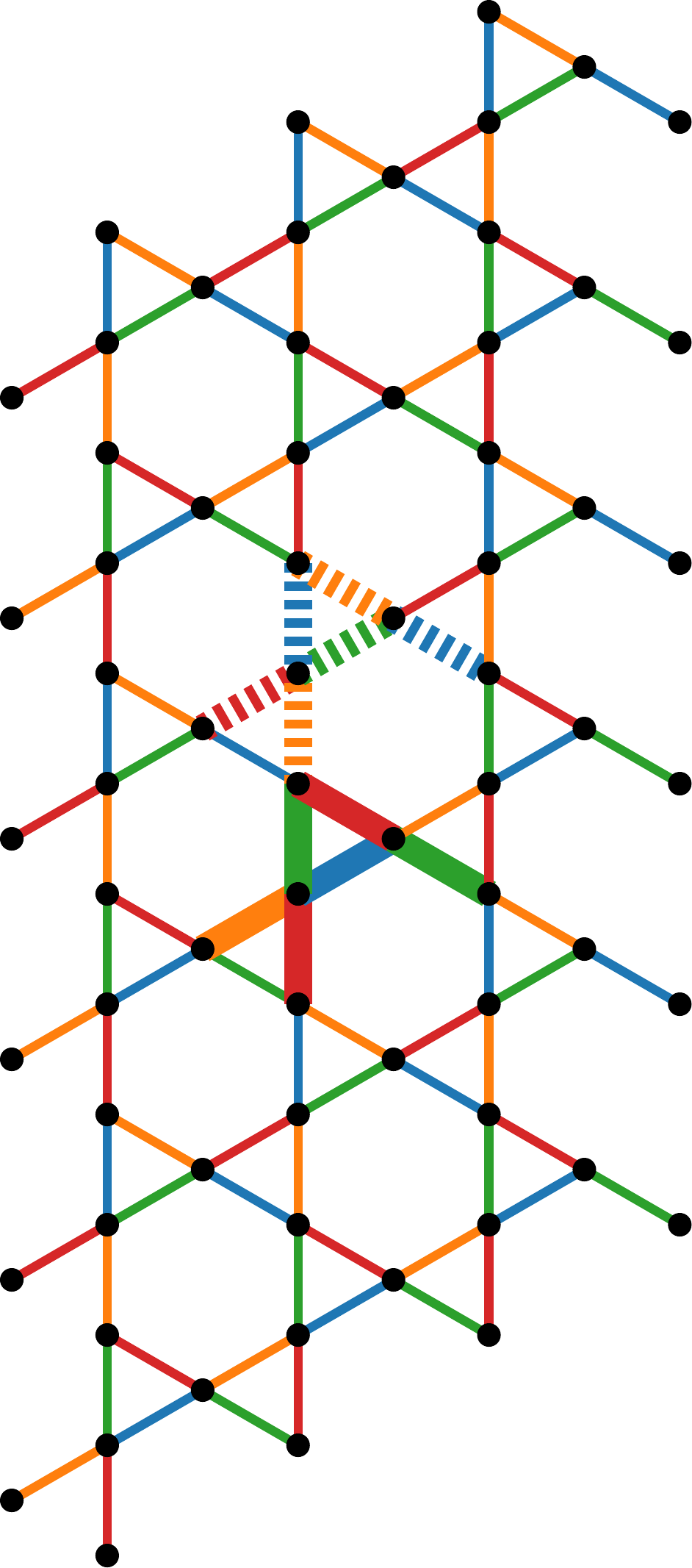}&\lf{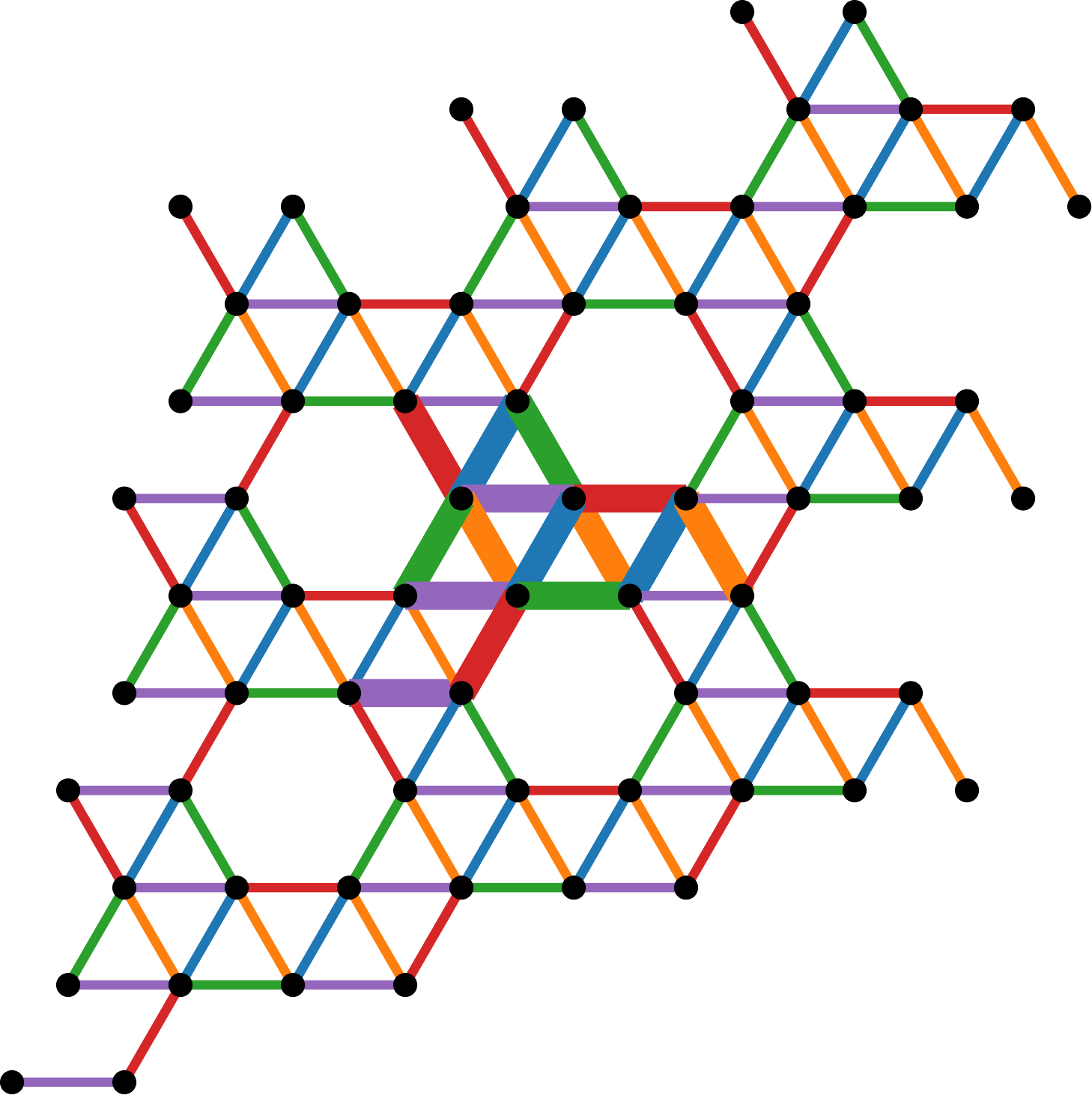}
    \\ \nm{$(3^3,4^2)$\\trellis\\elongated-triangular\\isosnub-quadrille}& \nm{$(3^2,4,3,4)$\\snub-square\\Shastry-Sutherland\\snub-quadrille}&
    \nm{$(3,6,3,6)$\\kagome\\ trihexagonal\\ hexadeltille}&
    \nm{$(3^4,6)$\\ bridge\\ maple-leaf\\ snub-trihexagonal\\ snub hextille}
  \end{tabular}

  \vspace{1em}
  \begin{tabular}{lll}
    \lff{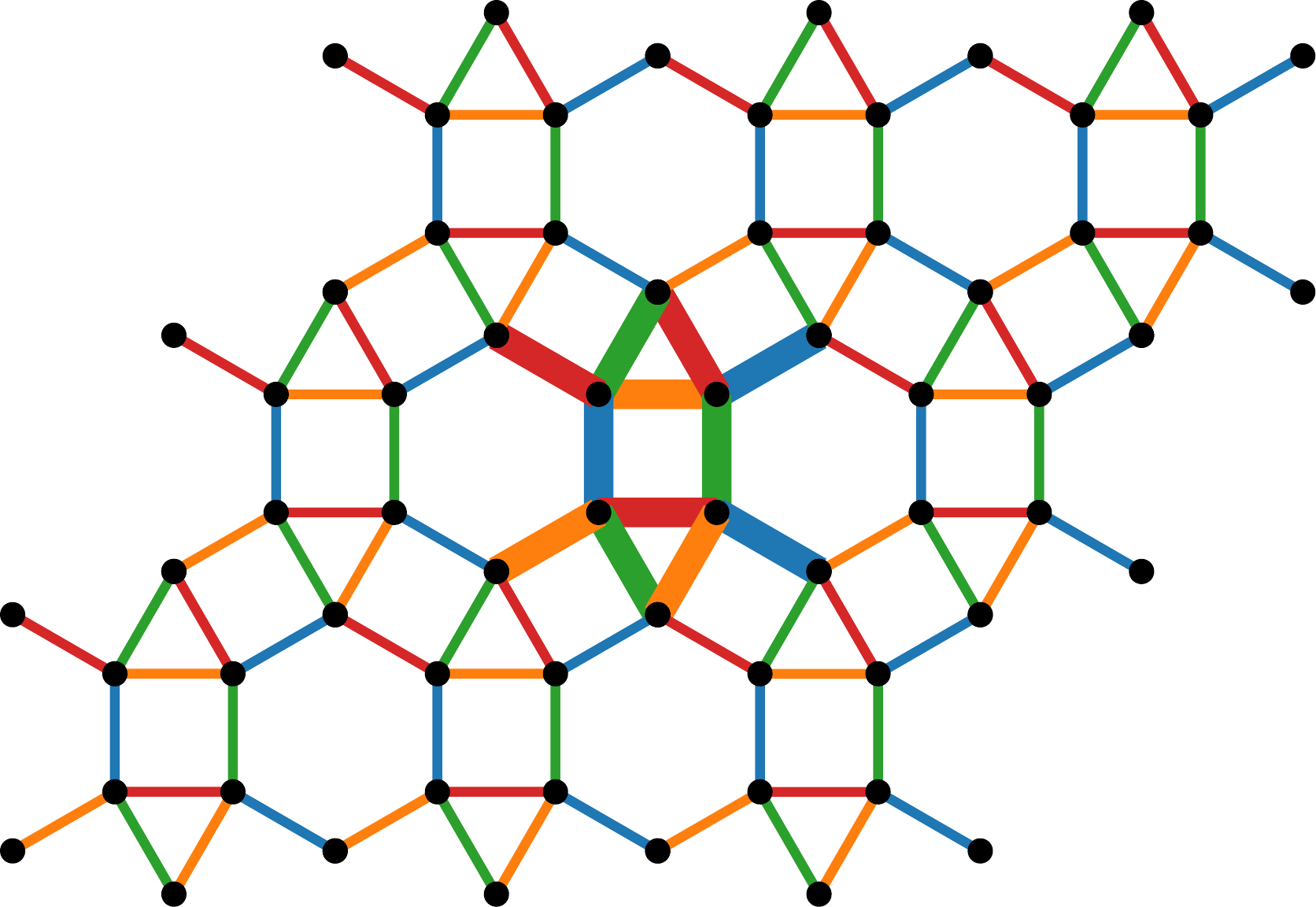}&
    \lff{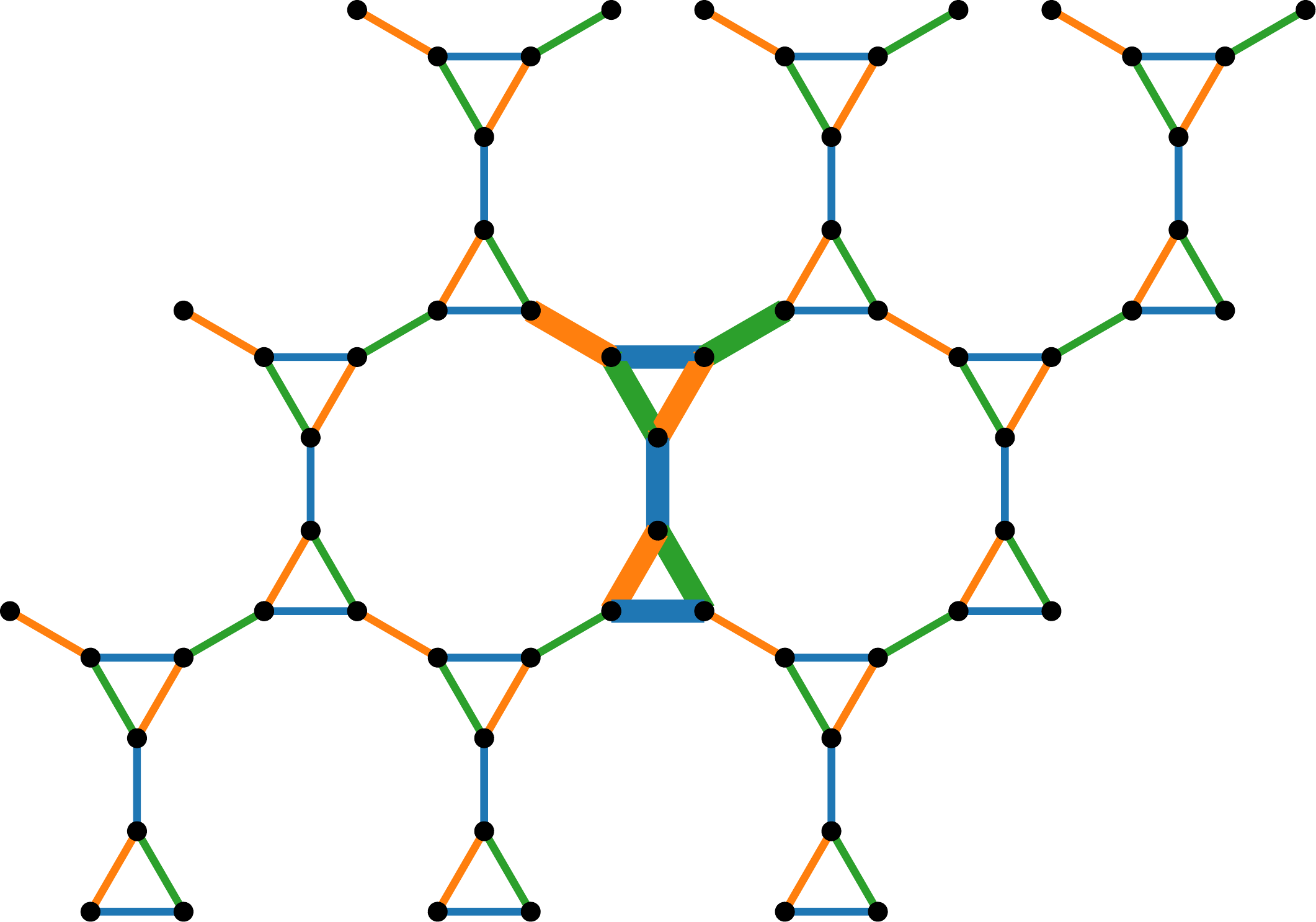} \hspace{3em} &
    \lff{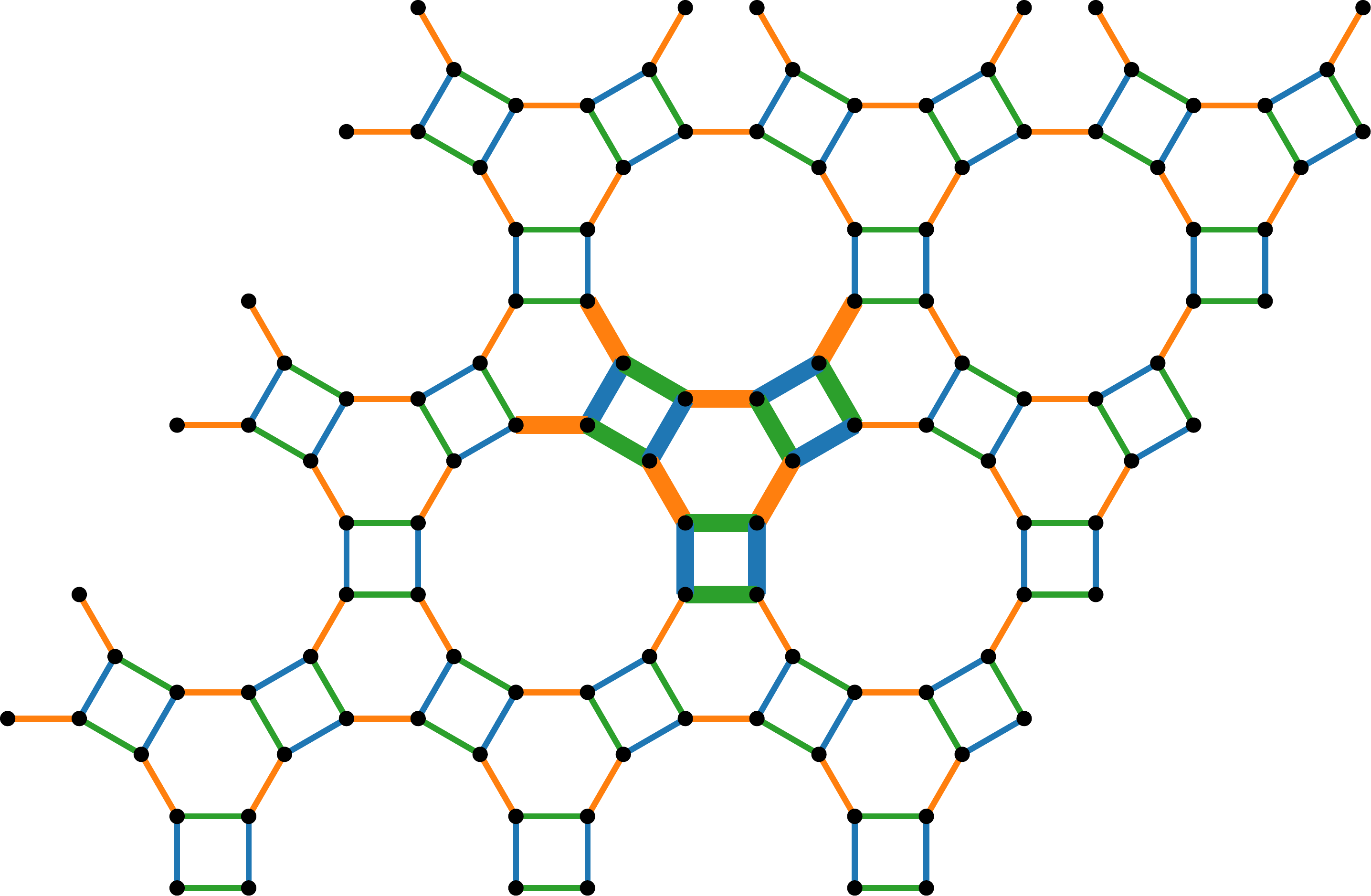}
    \\
    \nmm{$(3,4,6,4)$\\ ruby\\ bounce\\ rhombitrihexagonal\\ rhombihexadeltille}&
    \nmm{$(3,12^2)$\\ star\\ extended-kagome\\ truncated-hexagonal\\ truncated-hextille}&
    \hspace{2em}\nmm{$(4,6,12)$\\ cross\\ SHD\\ truncated-trihexagonal\\ truncated-hexadeltille}
  \end{tabular}
\caption{\label{fig:archimedean} Minimal edge colorings of the 11 Archimedean lattice graphs. The coloring basis graph, that is, the coloring pattern that needs to be repeated to color the entire lattice graph, is depicted with bold lines, both solid and dashed. It always contains the basis graph of the uncolored lattice graph (bold, solid lines only). Below the lattice graph follow various names of that graph. The first name is in the notation of Gr\"unbaum and Shephard~\cite{grunbaum2016tilings}. Starting at any vertex, it lists the number of vertices of each of the adjacent polygons in cyclic order. Any repetitions are abbreviated by a superscript, e.g., $(6^3)=(6,6,6)$. This name is followed by various other (possibly nonsystematic) names, whereas the last name is the one given by Conway et al.~\cite{conway2008symmetries}.}
\end{figure*}

\begin{figure*}
  \begin{tabular}{cccc}
    \lf{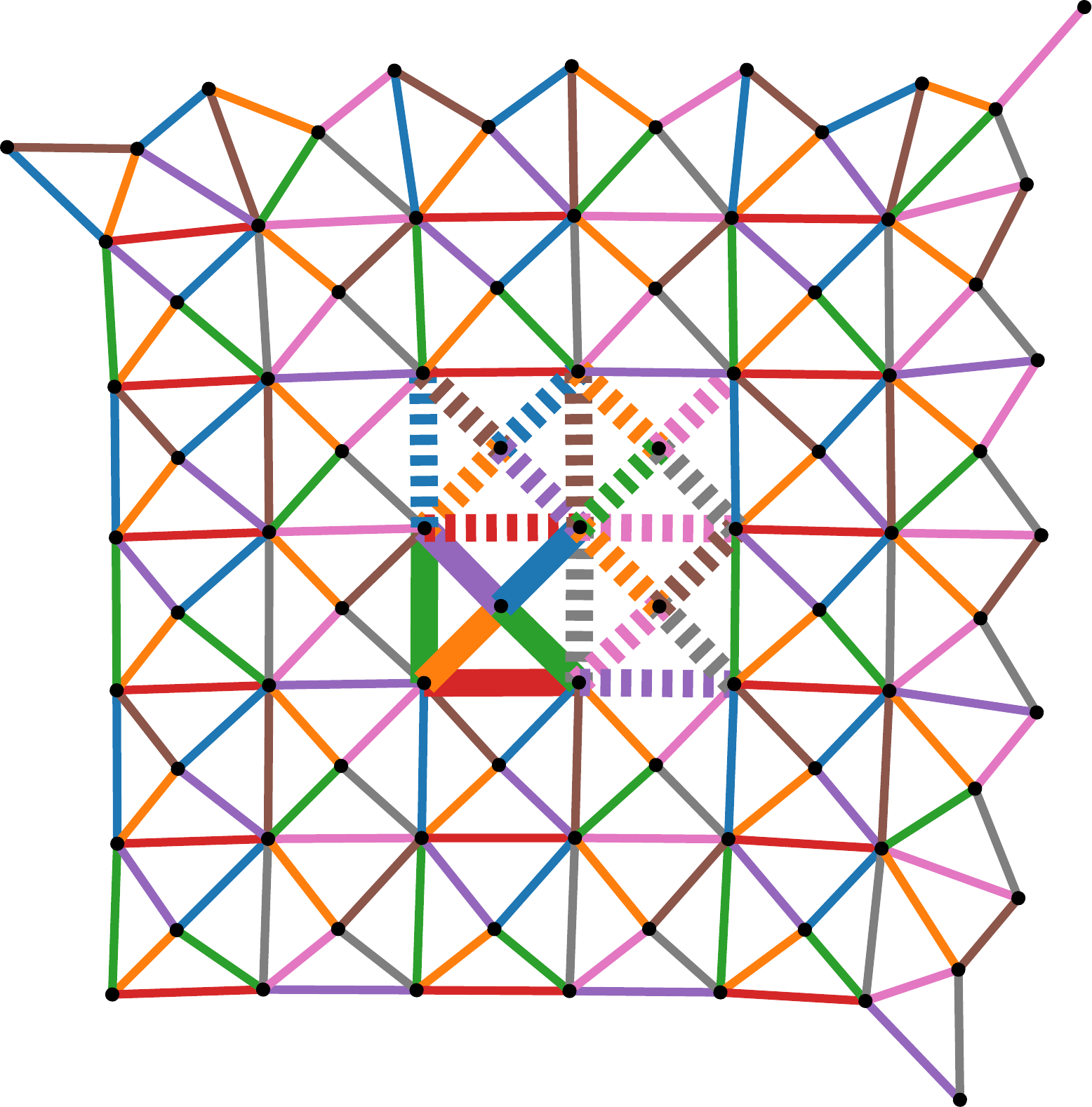}&
    \lf{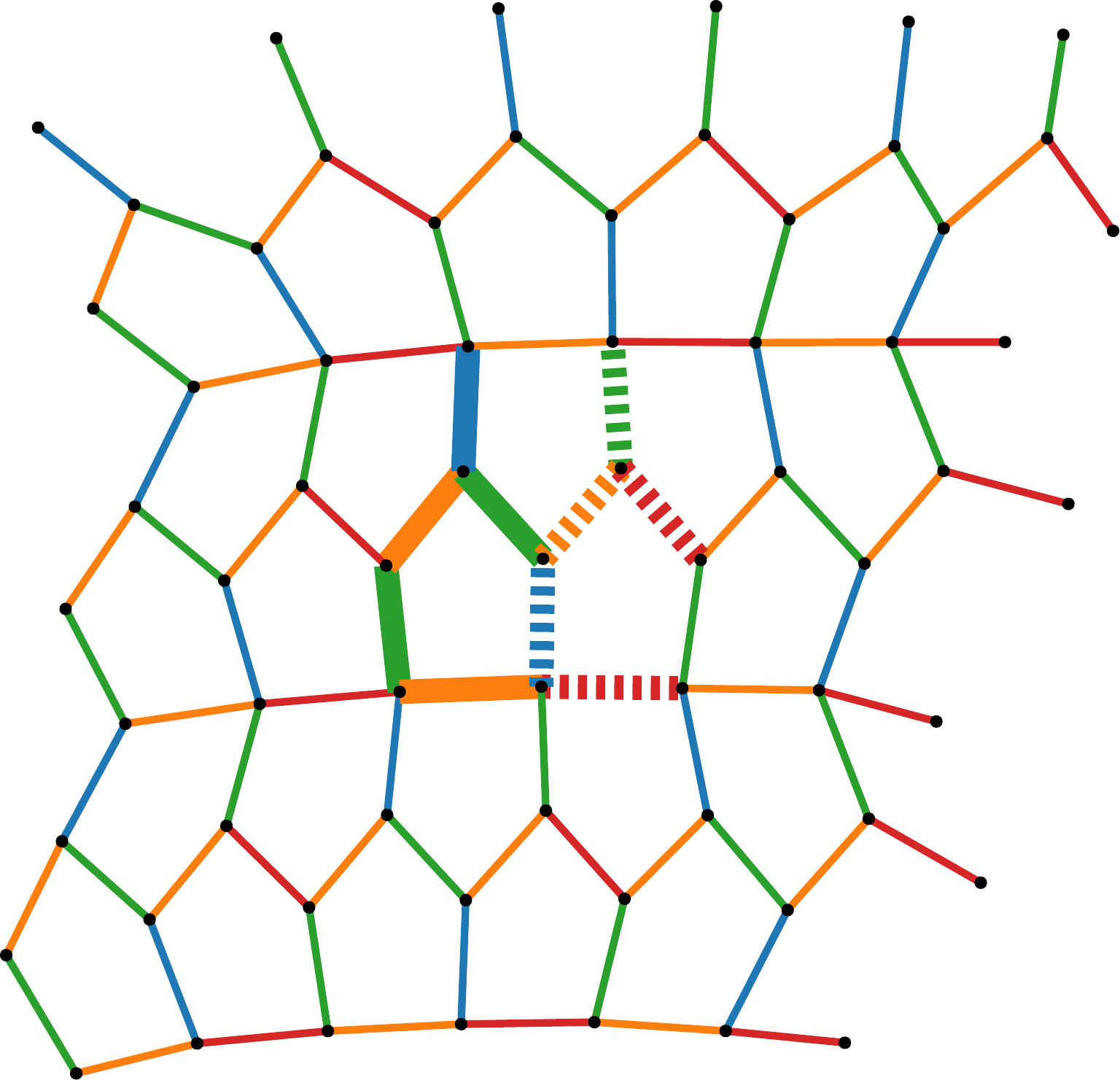}&
    \lf{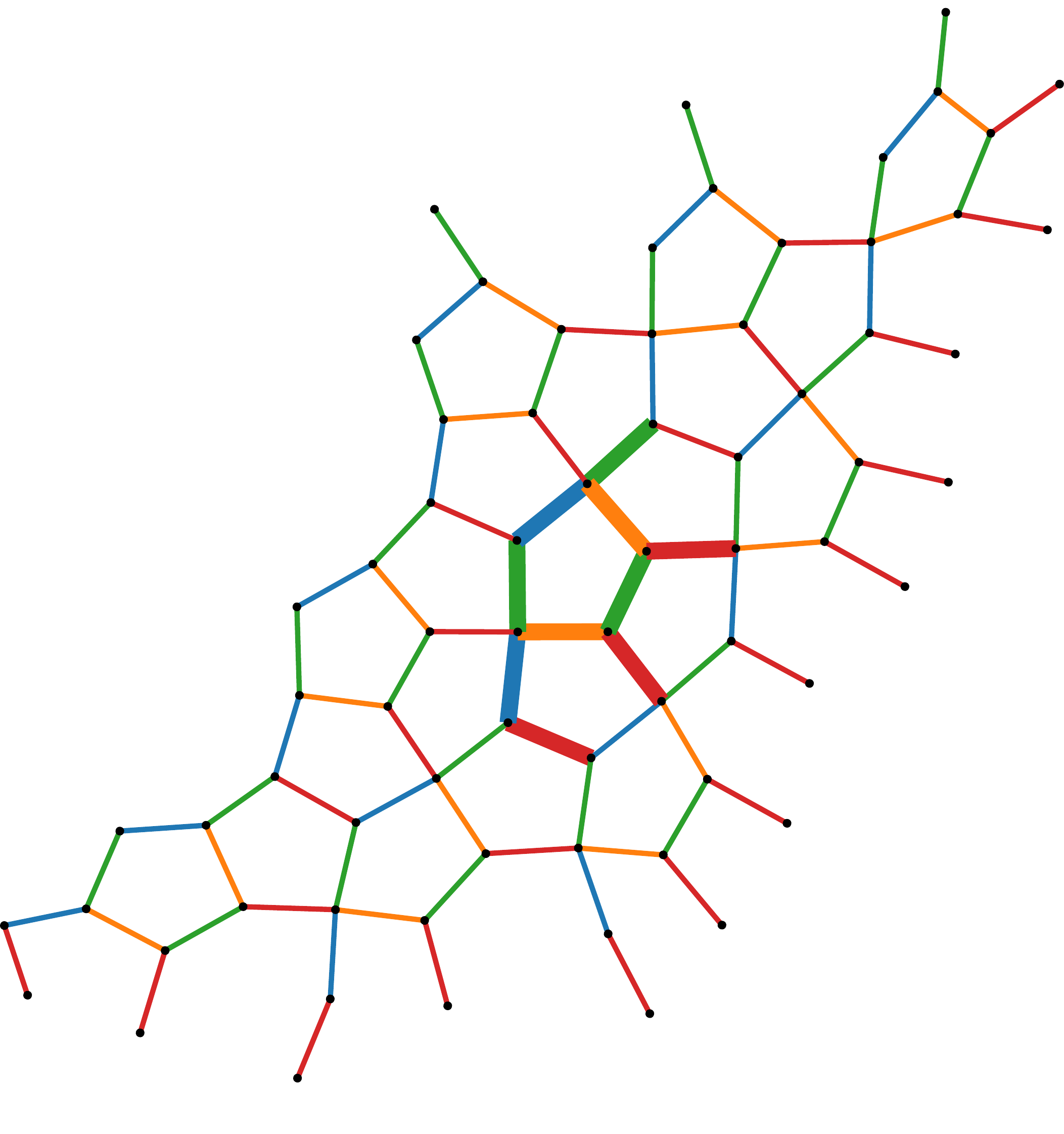}&
    \lf{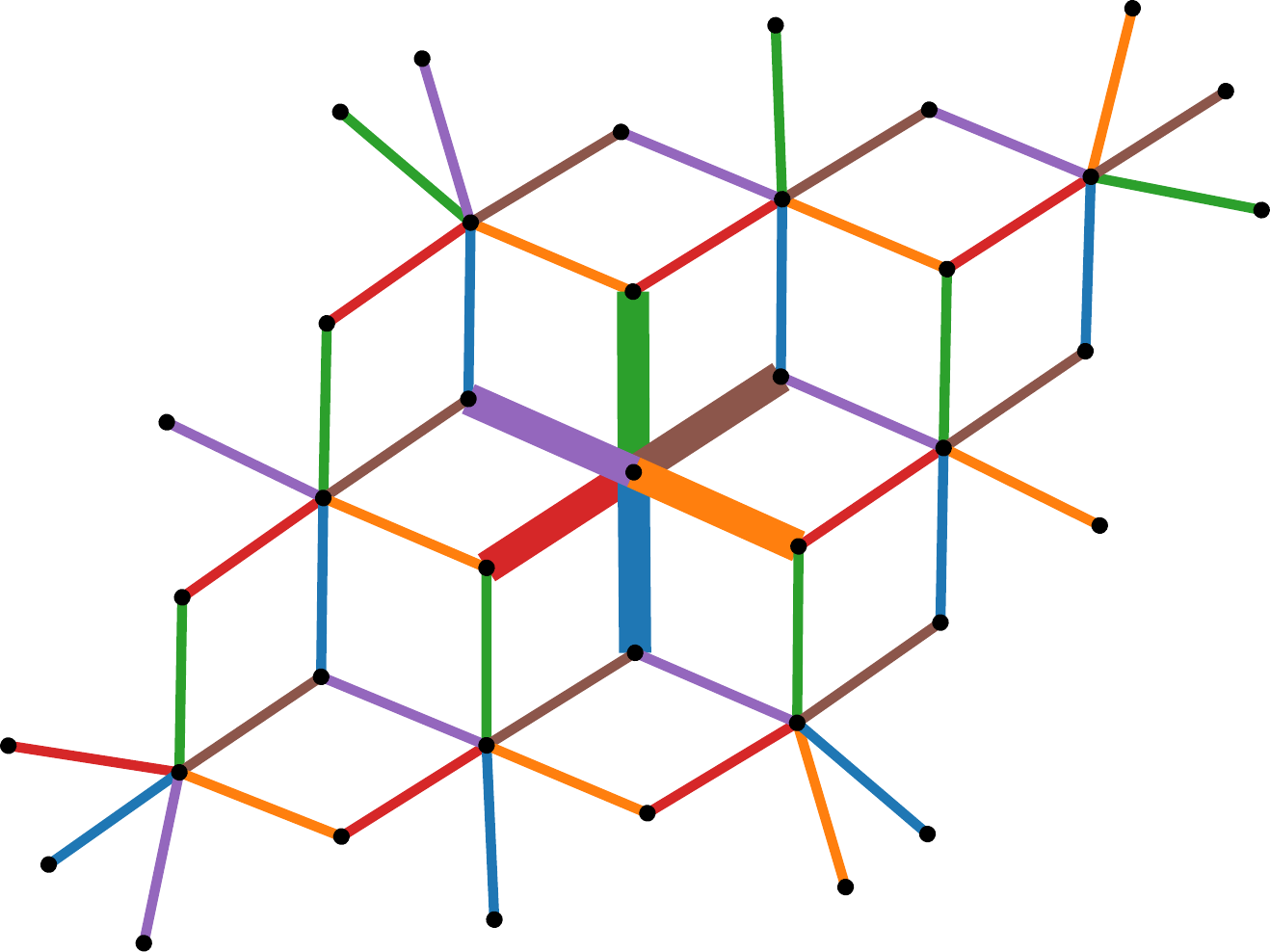}
    \\
    \nm{$[4,8^2]$\\ union jack}&
    \nm{$[3^3,4^2]$\\ prismatic-pentagonal\\ iso(4-)pentille}&
    \nm{$[3^2,4,3,4]$\\ cairo pentagonal\\ 4-fold-pentille}&
    \nm{$[3,6,3,6]$\\ dice\\ rhombille}
  \end{tabular}
  \begin{tabular}{cccc} \lf{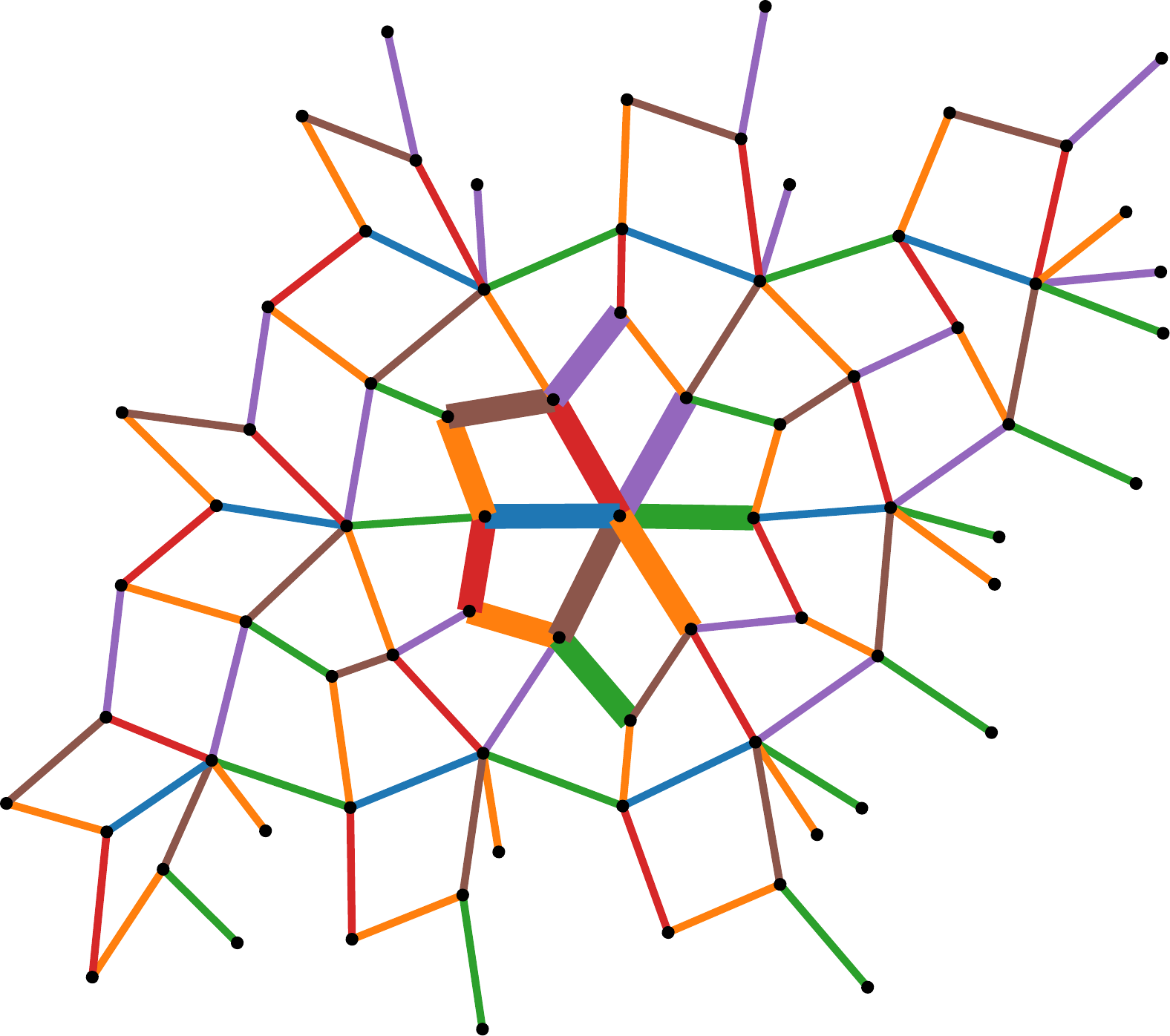}&\lf{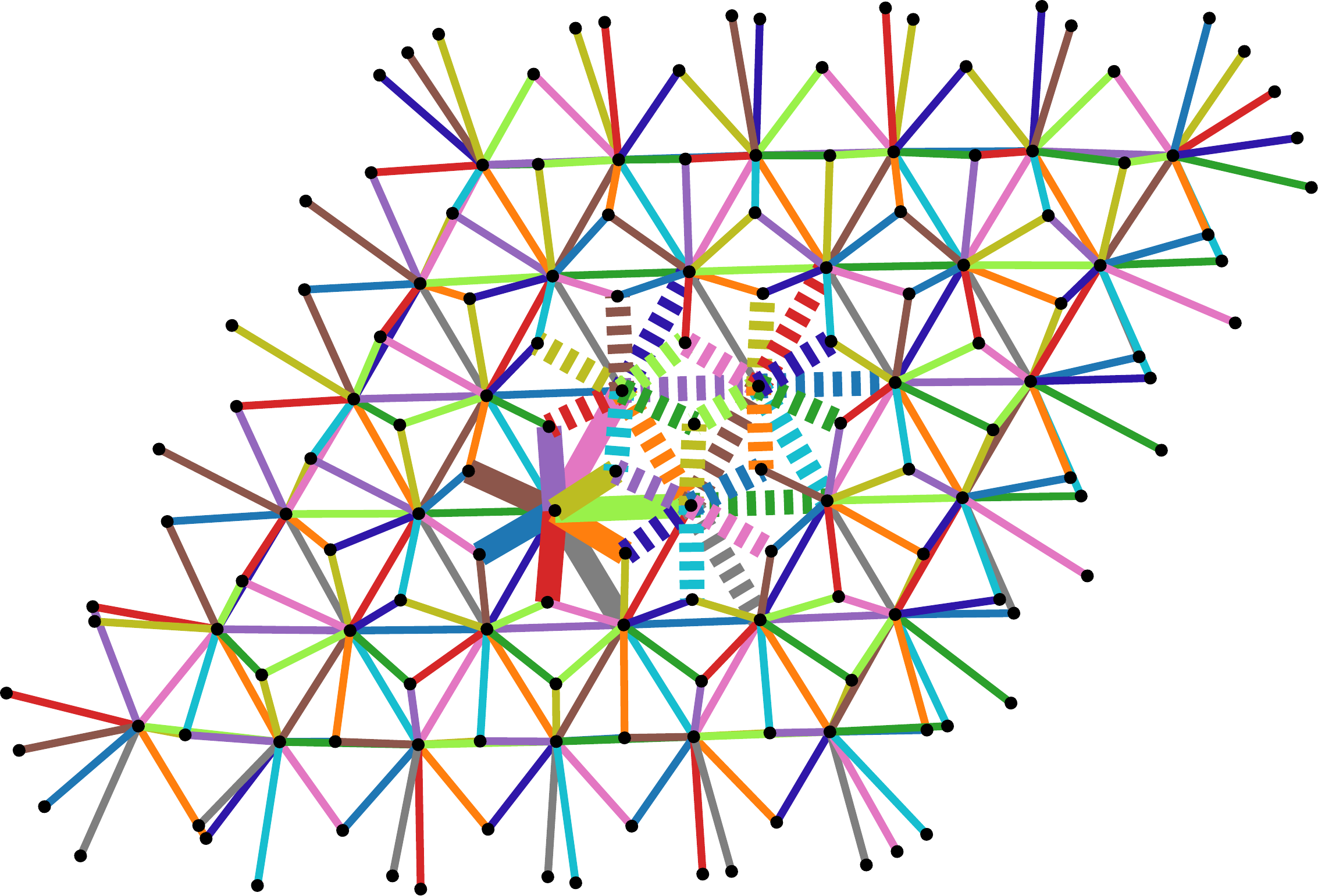}&\hspace{2em}\lf{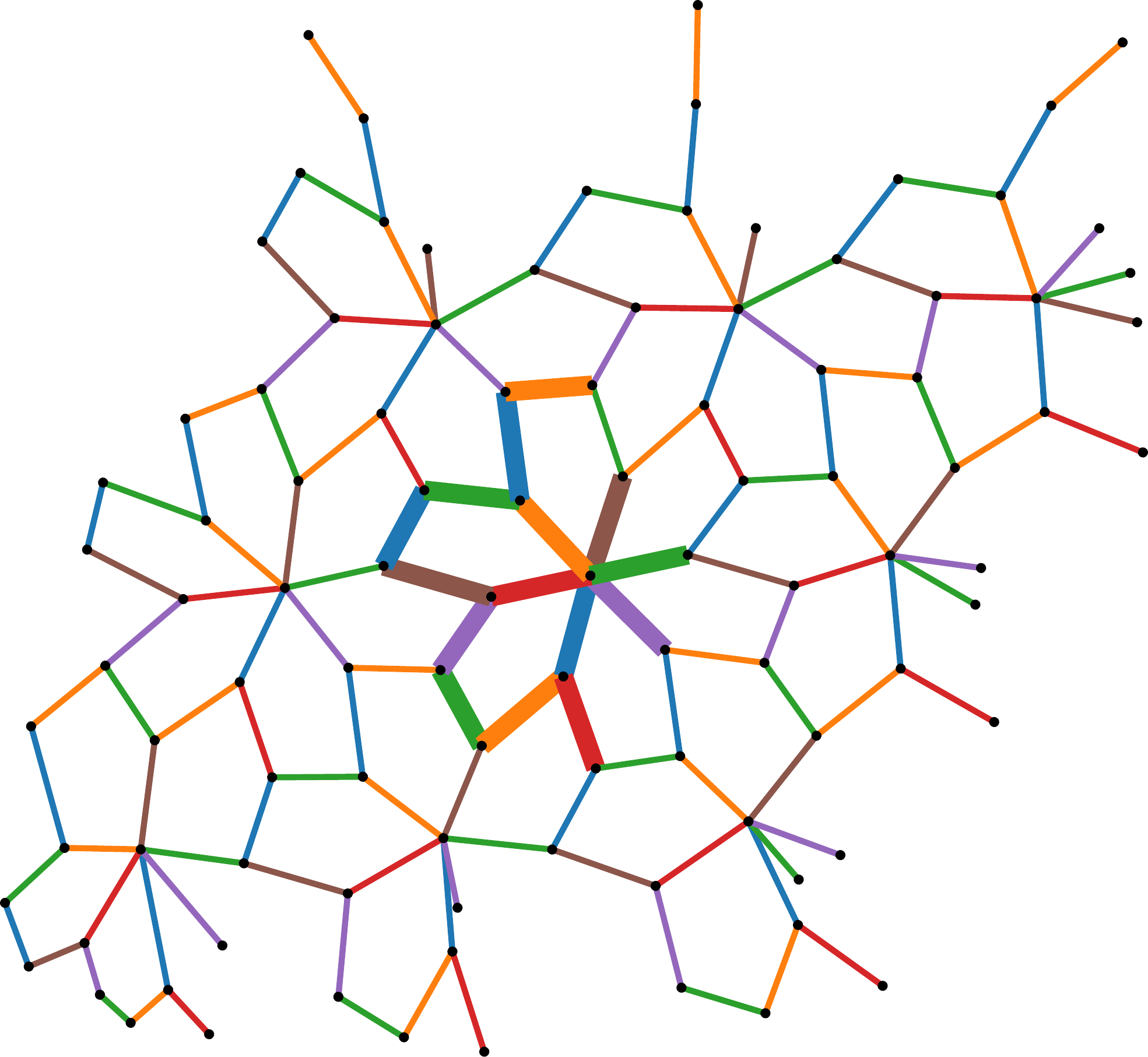}&\lf{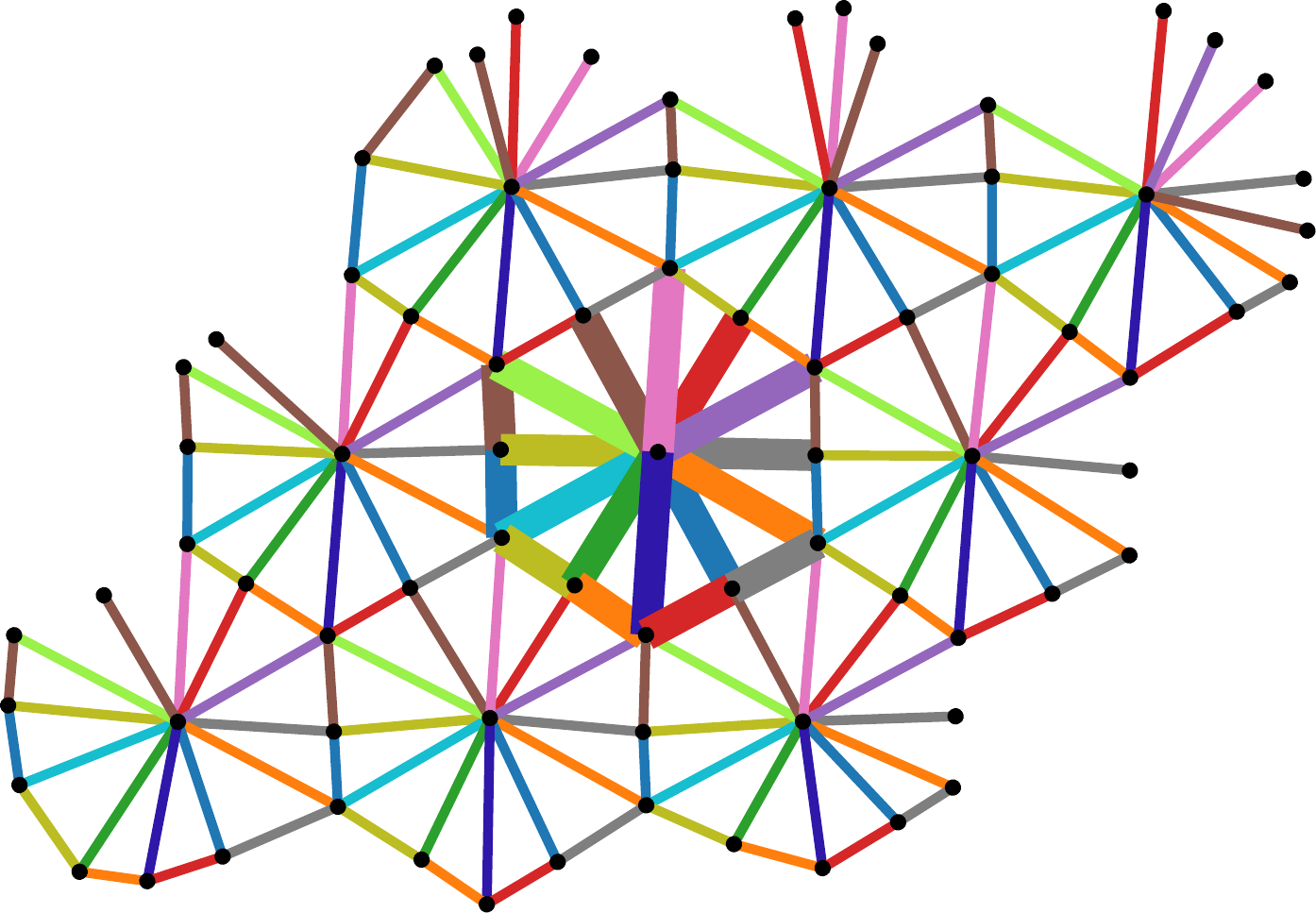}
    \\
    \nm{$[3,4,6,4]$\\ Deltoidal-trihexagonal\\ tetrille}&
    \nm{$[3,12^2]$\\ asanoha\\ hemp-leaf\\ Triakis-triangular\\ kisdeltille}&\hspace{2em}
    \nm{$[3^4,6]$\\ floret pentagonal\\ 6-fold pentille}&
    \nm{$[4,6,12]$\\ kisrhombille}
  \end{tabular}
  \caption{\label{fig:laves} Minimal edge colorings of all Laves tilings that are not already in \fig{archimedean}. The data are presented as in \fig{archimedean}, except for the first name, which is now the name of the Laves lattice by Gr\"unbaum and Shephard~\cite{grunbaum2016tilings}. Taking any tile, it lists (in square brackets) the degree of the tile's vertices in the lattice. The dual of $(x)$ is $[x]$.}
\end{figure*}

\begin{figure*} \includegraphics[width=.6\textwidth]{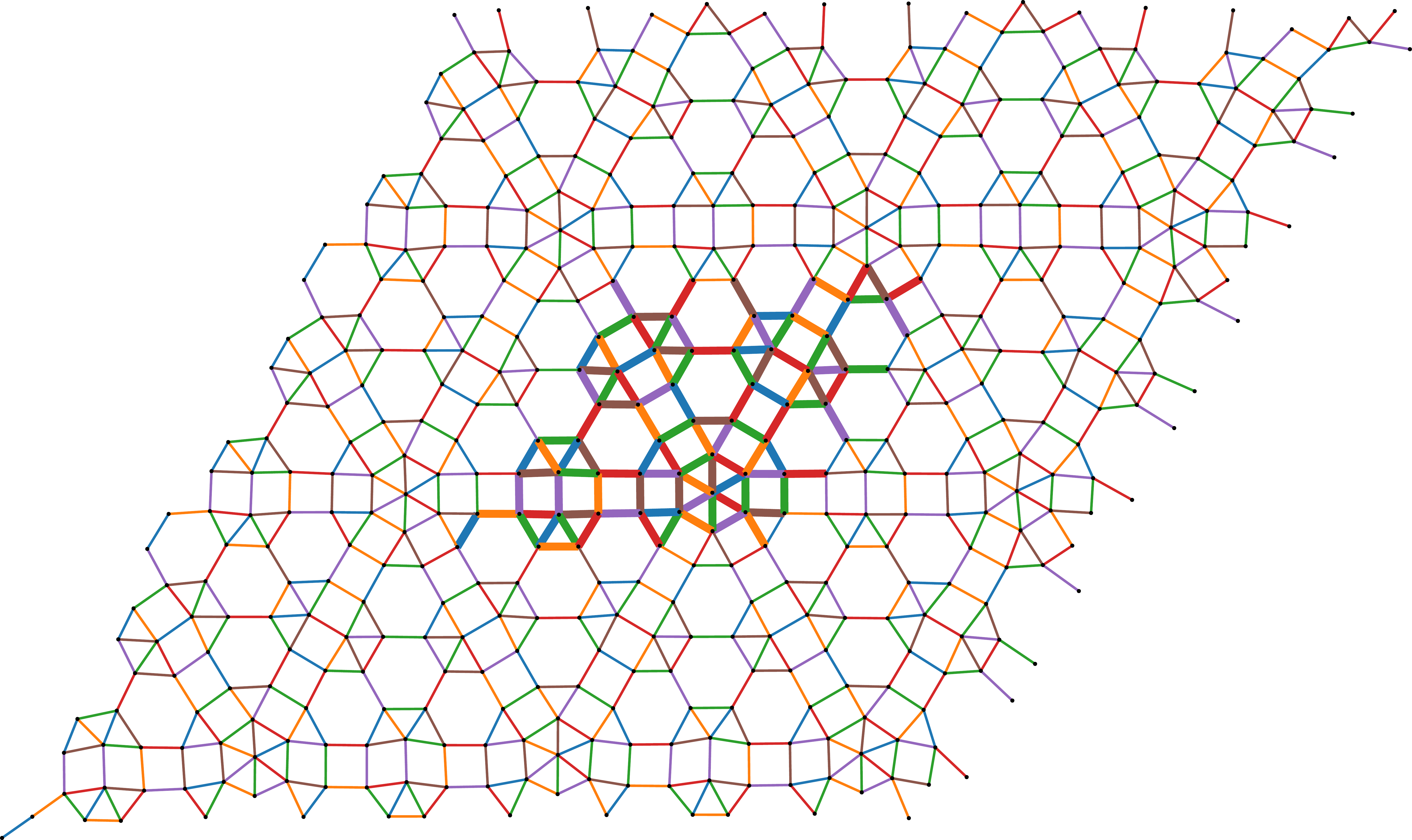}\\ \caption{\label{fig:6_uniform} Minimal edge coloring of a 6-uniform geometrical lattice graph from Refs.~\cite{galebach2023n_uniform,galebach_json}.}
\end{figure*}

\begin{figure*}
\begin{tabular}{cc}
     \lf{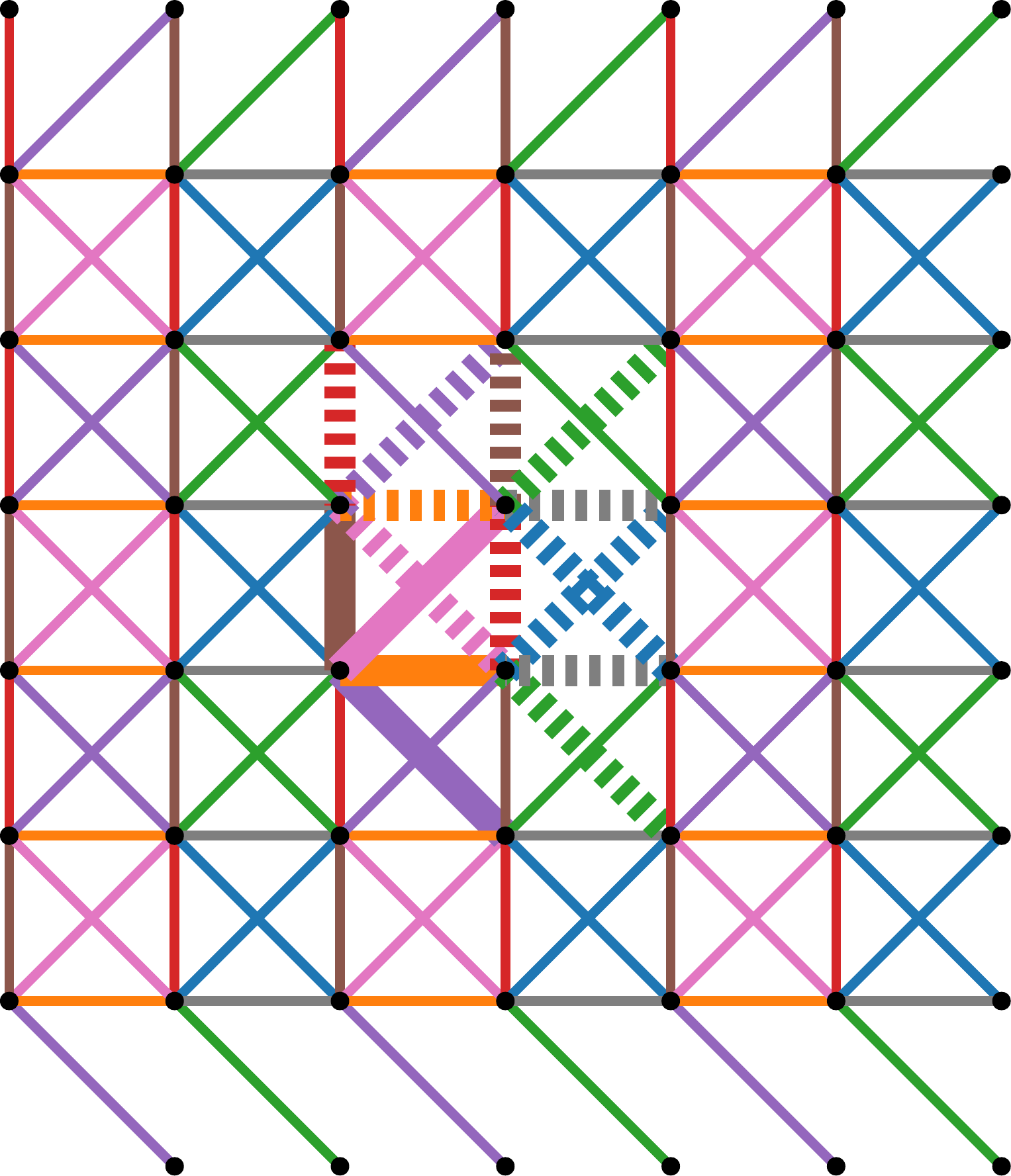}& \hspace{4em}
     \lf{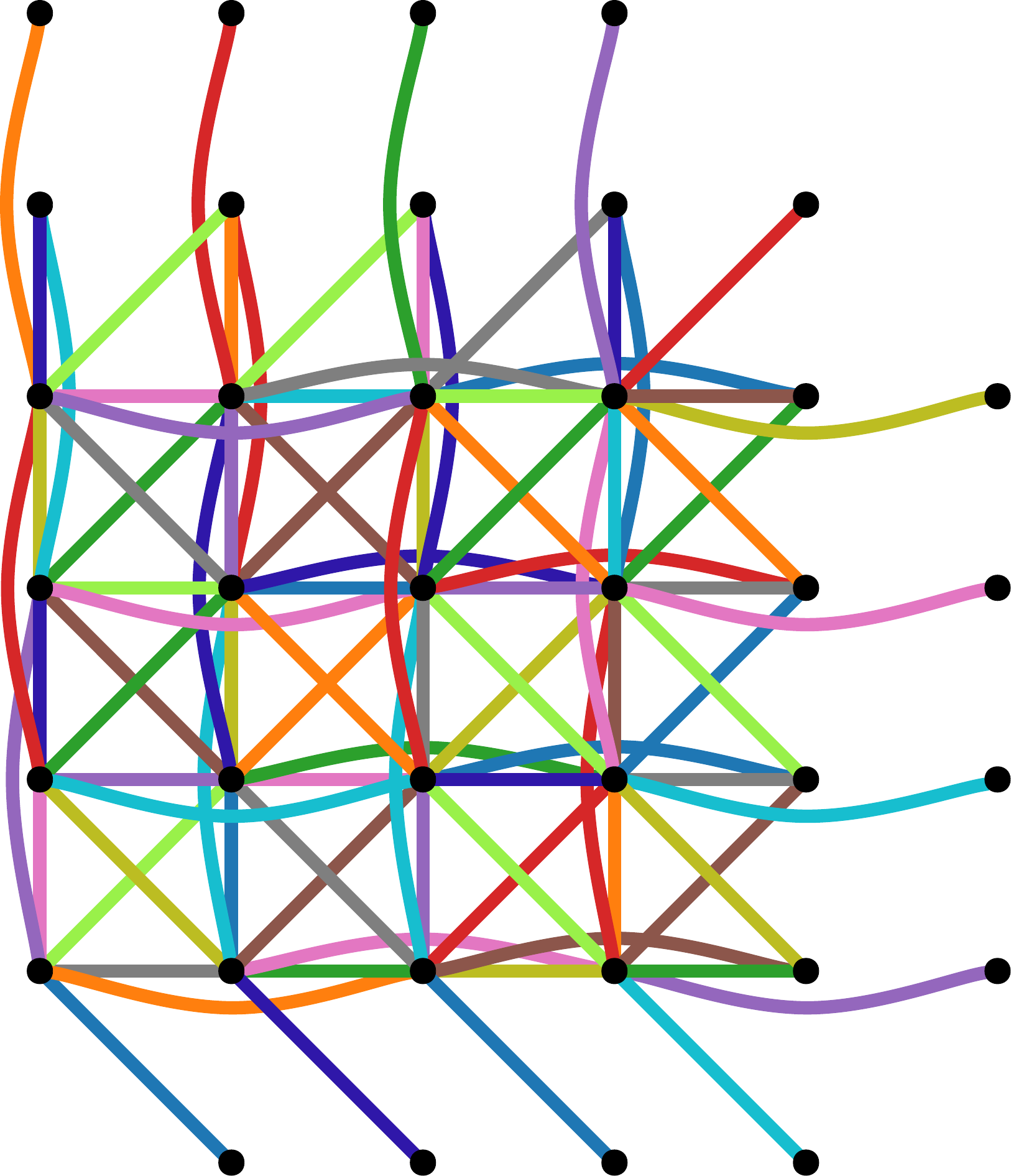}
     \\
     \nm{NN-square\\dense grid \\ $J_1J_2$-square \\ diagonal}&\hspace{4em}
     \nm{NNN-square\\ $J_1J_2J_3$-square}
\end{tabular}

\vspace{2em}
\begin{tabular}{ccc}
     \lf{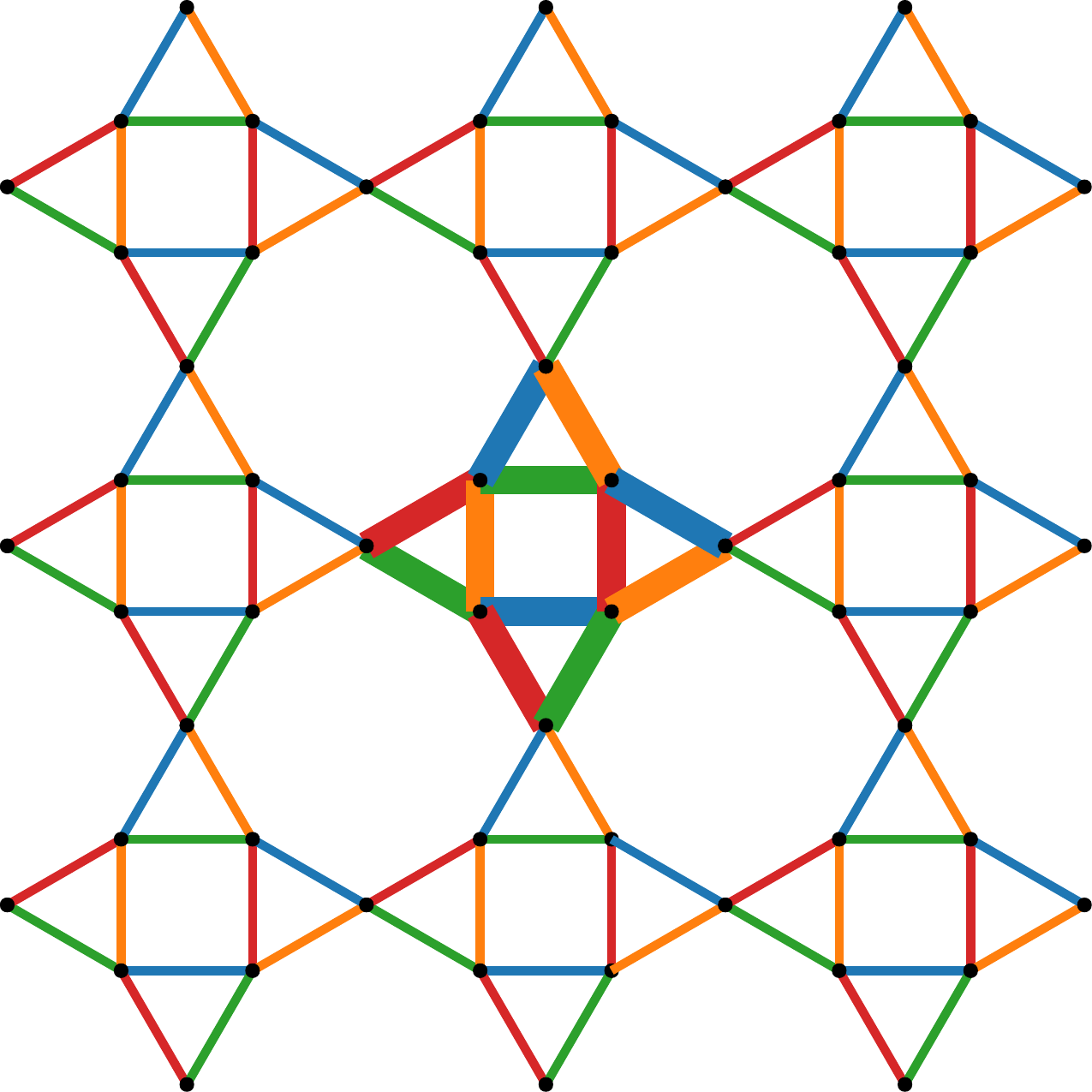}&
     \lf{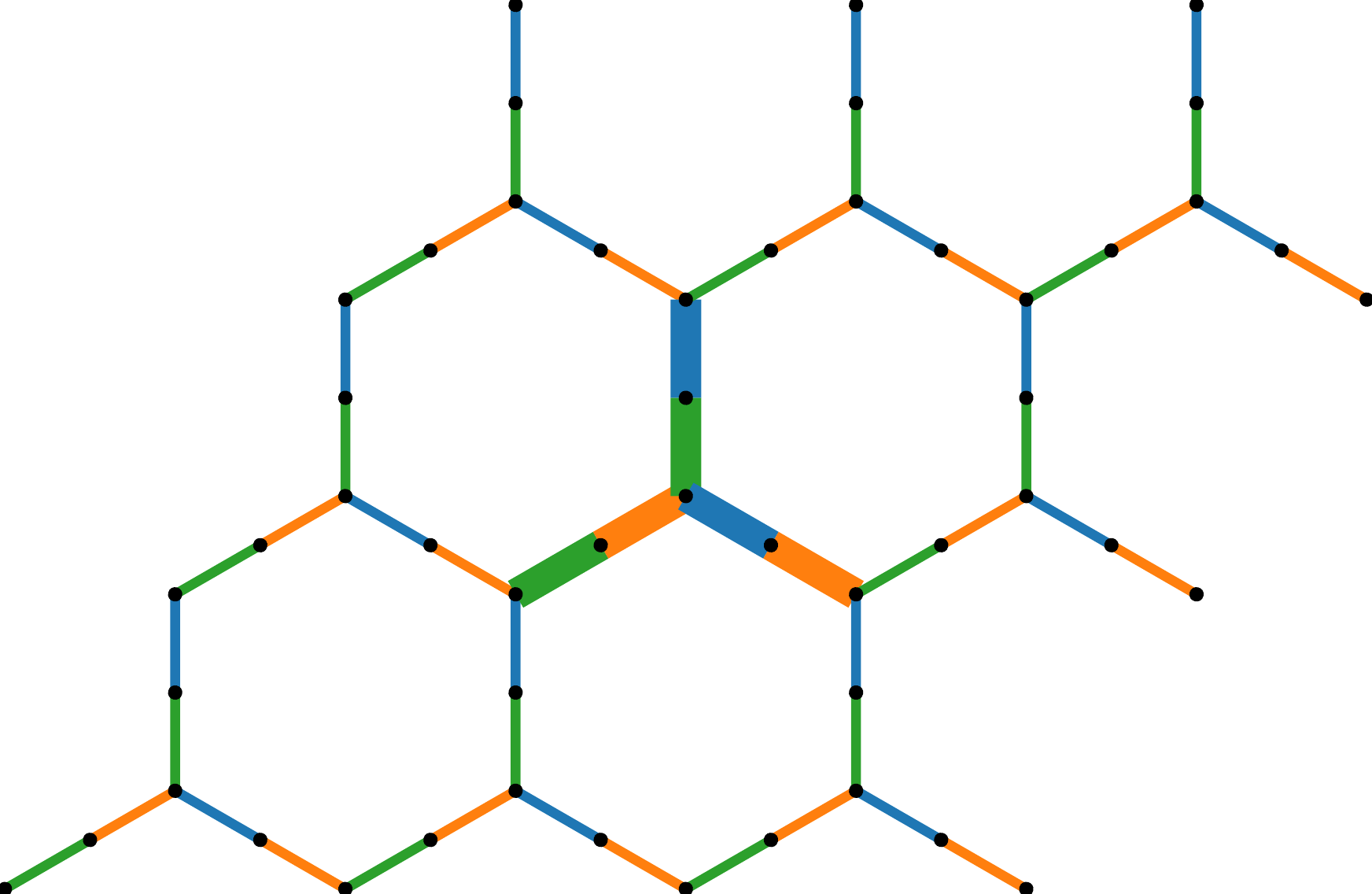}&\hspace{2em}
     \lf{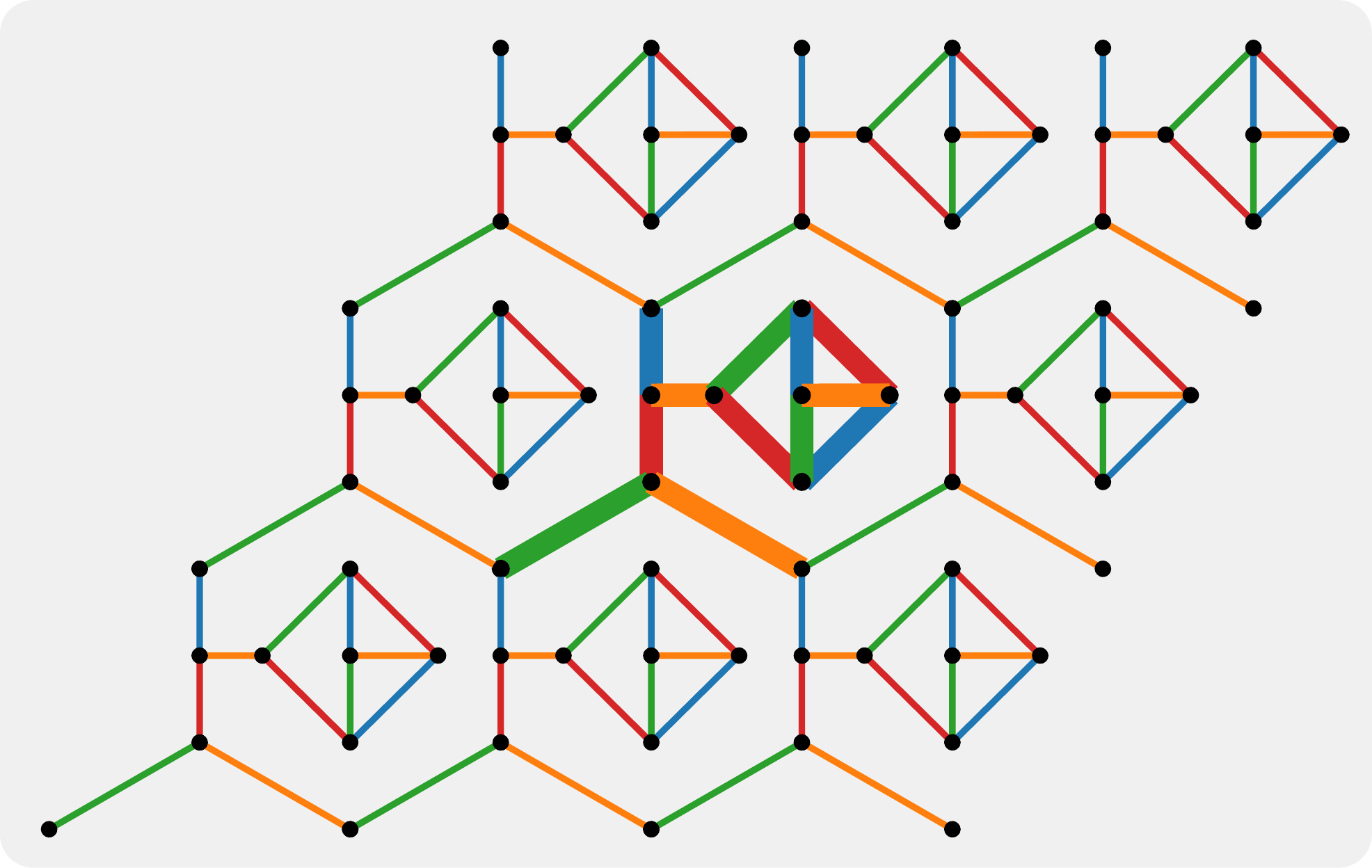}
     \\
     \nm{shuriken\\square-kagome\\squagome}&
     \nm{heavy-hex\\heavy-hexagon}&\hspace{2em}
     \nm{\ wheel-decorated honeycomb}
\end{tabular}
\caption{\label{fig:miscellaneous}
  Miscellaneous minimal edge colorings. In the case of the NNN-square lattice graph, to enhance legibility and unlike other lattice graph depictions, a single the coloring basis graph is presented, with all edges solid. The uncolored lattice is induced by the bottom-left node and its neighbors. A gray background was added to the wheel-decorated honeycomb lattice graph to accentuate its unique status as the only class~II lattice graph within this paper.}
\end{figure*}

\subsection{$k$-uniform lattices}
A tiling by regular polygons is called uniform if every vertex of the tiling can be mapped to any other vertex of the tiling by a lattice symmetry~\cite{grunbaum2016tilings}. Notice that this mapping involves lattice symmetries rather than isometries. Nevertheless, the set of uniform tilings of the plane by regular polygons is identical to that of the Archimedean lattices~\cite{conway2008symmetries}.

The notion of uniformity generalizes to $k$-uniformity. A tiling by regular polygons is called $k$-uniform if the vertices fall into $k$ classes, where any two vertices within a class can be mapped to one another by a symmetry of the tiling~\cite{grunbaum2016tilings, krotenheerdt1968homogenen}. In this terminology, the 11 Archimedean (or uniform) tilings are 1-uniform. There are 20 2-uniform tilings of the plane~\cite{krotenheerdt1968homogenen} and 61 3-uniform tilings~\cite{chavey1989tilings}.

For $k\leq 6$, Galebach~\cite{galebach2023n_uniform} performed an exhaustive search for all $k$-uniform tilings using a computer program. For $k=1,\ldots,6$, he found the number of $k$-uniform tilings to be 11, 20, 61, 151, 332, 673, reproducing previous results for $k\leq 3$. Although no scientific description of Galebach's method exists, his list is considered exhaustive by some authors~\cite{oeisA06859,weisstein2023uniform}. Čtrnáct~\cite{ctrnact2023catalog} performed an exhaustive search for $k\leq 12$, reproducing Galebach's results for $k\leq 6$. Unfortunately, also for Čtrnáct's catalog, a scientific description of the used methods is lacking. A scientific description of a method for enumerating all $k$-uniform edge-to-edge tilings by regular polygons is introduced in Ref.~\cite{gomez2021gomjau}, but this method has not yet been used to generate tilings for $k\geq 3$.

In Refs.~\cite{sanchez2019acquiring,sanchez2021integer,medeiros2018synthesizing}, Galebach's catalog of $k$-uniform tilings was converted to a notation where each tiling is represented by its seeds, where every seed is a point in the complex plane. The catalog in the latter form was obtained at \cite{galebach_json} and is repeated in the Supplemental Material \cite{Note2} for persistency. We automatically converted all entries in the catalog to the basis graph notation.

We found a type-I and hence minimal edge-coloring for all 1248 $k$-uniform ($k\leq 6$) tilings in the catalog in 10 minutes. Images of these edge colorings are available as Supplemental Material~\cite{Note2}. An example minimal edge-coloring of a 6-uniform tiling is shown in \fig{6_uniform}.

\subsection{Miscellaneous lattice graphs}
For the applications in \sec{previous_work_and_applications}, miscellaneous lattice graphs are of interest that are not among the Archimedean, Laves or $k$-uniform lattices. We consider the square lattice with added edges to the geometrically next-nearest neighbors (NN-square) and the NN-square lattice graph with added edges to geometrically next-next-nearest neighbors (NNN-square). As the kagome lattice graph, the shuriken lattice graph consists of corner-sharing triangles. The heavy-hex lattice graph is obtained from the hexagonal lattice graph by making the edges `heavy', that is, by adding a vertex on top of every edge of the hexagonal lattice graph. Type-I colorings of these four graphs are displayed in \fig{miscellaneous}. The wheel-decorated honeycomb lattice is provably class~II (\sec{analysis}). Figure \ref{fig:miscellaneous} also shows a type-II (and hence minimal) edge coloring of this graph. The five edge colorings in \fig{miscellaneous} were found in 25.6\,s, of which the NNN-square lattice graph is the most challenging, taking 25.5\,s.

In addition to the lattice graphs in the Galebach collection,  Refs.~\cite{sanchez2019acquiring,sanchez2021integer,medeiros2018synthesizing} give lattice graphs in the Sa \& Sa collection~\cite{sa2017sobre}, which we obtained at Ref. \cite{saesa_json}. These graphs were created for artistic purposes and do not intend to exhaust any class of graphs. However, the application of \met{edge_coloring} to this extra collection further establishes the practical applicability of the edge coloring method. A type-I and hence minimal edge coloring was found for all 57 lattice graphs in the Sa \& Sa collection that were not already in the Galebach collection, taking 24\,s. The results are not depicted in this paper, but are available in the Supplemental Material~\cite{Note2}.

\section{Analysis}\label{sec:analysis}

We show that the self-loop-free property of a wrapped patch of a lattice graph $\mc G$ is necessary and sufficient for this patch to induce a proper edge coloring of $\mc G$ after it is properly edge colored and unwrapped (\thm{proper_edge_coloring}). This is used to show that \met{edge_coloring} halts and that it is correct if we restrict the input to $t=3$ (i.e., `find a proper edge coloring') or $t=2$ (i.e., `find a type-II coloring'). Correctness is also shown if $t=1$ (i.e., `find a type-I coloring'), provided that the method halts (\thm{correctness}). We analyze the method's running time in the cases where it is guaranteed to halt as formalized in \alg{edge_coloring_algorithm}. We prove by construction that not all lattice graphs, and not even all planar lattice graphs, are class~I (\pro{class-II}). This shows that there are problem instances where \met{edge_coloring} does not halt in case a type-I coloring is requested. Finally, we discuss the relation of the method to the vertex coloring of lattice graphs. Accompanying lemmas and proofs are deferred to \app{proofs}.

\subsection{Correctness}\label{sec:correctness}
In the following theorem and lemma, let the basis graph $B$ generate the lattice graph $\mc G$. Let $P_{n,m}$ be a patch of $\mc G$ of $n$ by $m$ basis graphs, $\tilde P_{n,m}=W(P_{n,m})$ the wrapped patch, and $\tilde C$ a proper edge coloring of the wrapped patch. Let $C=U(\tilde C)$ be the unwrapped $\tilde C$. Let $\mc C$ be the edge colored lattice graph induced by $C$. It forms an edge coloring of $\mc G$. The following theorem shows how \met{edge_coloring} hinges on stage~2 (`no self-loops').

\begin{theorem}\label{thm:proper_edge_coloring}
  The edge colored lattice graph $\mc C$ is properly edge colored if and only if the wrapped patch $\tilde P_{n,m}$ is self-loop-free.
\end{theorem}

Any graph can be properly edge colored and this coloring may use any number of colors. Thus, to show that \met{edge_coloring} halts and produces the correct output in case a proper edge coloring is requested, by \thm{proper_edge_coloring}, it suffices to show that in \met{edge_coloring}, a self-loop-free patch is eventually encountered. We do this in \lem{self-loop-free} of \app{proofs}.

For a type-II coloring of a lattice graph $\mc G$, at most $\Delta(\mc G)$ colors may be used. The degree of the lattice graph may be obtained by computing the degree of a wrapped patch, provided that that patch is self-loop free.
\begin{lemma}\label{thm:max_degree}
   If the wrapped patch $\tilde P_{n,m}$ is self-loop-free, then $\Delta(\tilde P_{n,m})=\Delta(\mc G)$.
\end{lemma}
By Vizing's theorem, every finite simple graph can be type-II colored. Also note that for a colored and wrapped patch, unwrapping does not change the number of colors. Thus, if a wrapped patch is encountered by \met{edge_coloring} that is simple, it can be type-II colored and unwrapped, after which \thm{proper_edge_coloring} and \lem{max_degree} guarantee that this patch induces a type-II coloring of the lattice graph. Therefore, to show that \met{edge_coloring} halts and that it is correct in case a type-II coloring is requested, it only remains to be shown that a simple patch is eventually encountered by \met{edge_coloring}, which we do in \lem{simple} of \app{proofs}.

Finally, in case a type-I coloring is requested,  \thm{proper_edge_coloring} and \lem{max_degree} also provide that if a self-loop-free wrapped patch is encountered that allows a type-I coloring, then this patch induces a type-I coloring of the lattice graph after type-I coloring and unwrapping. It is, however, no longer guaranteed that such a patch is encountered.

Let $B,t$ be the input of \met{edge_coloring}. The above results are summarized in the following theorem (see \app{proofs}).
\begin{theorem}\label{thm:correctness}
If $t=2$ or $t=3$, \met{edge_coloring} terminates and the output $C,(n,m)$ induces a type-$t$ coloring of $\mc G$. If $t=1$ and \met{edge_coloring} terminates, the output induces a type-I coloring of $\mc G$.
\end{theorem}

\subsection{Running time}\label{sec:running_time}
We introduced \met{edge_coloring} as a \emph{method} to allow flexibility in its usage. For a rigorous worst-case running time analysis, we remove this flexibility and consider the following use case of the method. The method then becomes an algorithm in the sense that it is concrete, that it halts, and that it is correct.

The algorithm receives a basis graph $B$ (that induces a lattice graph $\mc G$) as input. For simplicity of the analysis, we assume that $B$ does not have isolated vertices (which play no role in the edge coloring). The algorithm outputs a coloring basis graph $C$ [together with its size $(n,m)$] that induces a type-II edge coloring of $\mc G$. Notice that any type-II coloring is also proper edge coloring, so that the algorithm can also be used when a mere proper edge coloring is desired. For a basis graph $B$, define
\begin{equation}
  \begin{aligned}\label{eq:max}
  dx_{\max}&=\max\{\lvert dx \rvert\mid(dx,dy,s)\in E(B)\},\\ dy_{\max}&=\max\{\lvert dy \rvert\mid(dx,dy,s)\in E(B)\}.
  \end{aligned}
\end{equation}

\begin{algorithm}\label{alg:edge_coloring_algorithm}
  Compute $dx_{\max},dy_{\max}$. Run \met{edge_coloring} with $t=2$ and the initial patch size $(n,m)$ set to $(2dx_{\max}+1,2dy_{\max}+1)$. Use the Misra and Gries edge coloring algorithm to edge color the finite patch in stage~3.
\end{algorithm}

In \app{proofs}, we prove that the algorithm is finite and correct. We also prove running time of \alg{edge_coloring_algorithm} is dominated by the running time of the Misra and Gries edge coloring algorithm, which is $O(\lvert E(P_{n,m}) \rvert^2)=O[(dx_{\max}dy_{\max}\lvert E(B) \rvert)^2]$. Define the maximum distance $D=\max\{dx_{\max},dy_{\max}\}$ and $\mu=\lvert E(B)\rvert$. Note that with this definition, for example, the distance between unit cell $(0,0)$ and $(1,1)$ (see \fig{lattice_graph}) is $D=1$. With these definitions, the worst-case running time of \alg{edge_coloring_algorithm} is $O(D^4\mu^2)$.

\subsection{Class II lattice graphs}\label{sec:class_II_lattice_graphs}

\begin{figure}
  \includegraphics[width=.2\columnwidth]{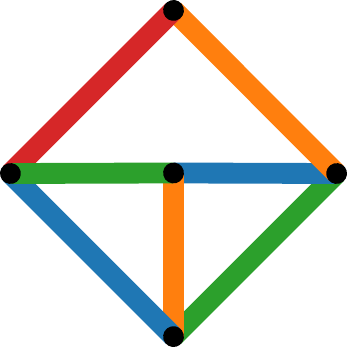}
  \caption{\label{fig:broken_wheel} Minimal edge coloring of the broken wheel graph $W_5'$. This graph is class~II.}
\end{figure}

We now show the existence of lattice graphs that are class~II. We define the broken wheel graph $W_5'$ as the wheel graph $W_5$ with one missing spoke (\fig{broken_wheel}). Beginning by coloring the edges incident on a vertex of degree three, it is straightforward to see that $W_5'$ is class~II, that is, that at least four colors are required to properly edge color~$W_5'$.

We define the \emph{wheel-decorated honeycomb lattice graph} by its basis graph as follows. Take the basis graph of the honeycomb lattice graph (\fig{lattice_graph}), add a vertex $v$ on top of one of the edges, add the wheel graph $W_5$, and redirect one of the spokes of $W_5$ so that it ends on $v$ instead of the hub (the central vertex) of $W_5$. The wheel-decorated honeycomb lattice graph is a simple, 3-regular planar graph and is depicted in \fig{miscellaneous}.

\begin{proposition}\label{thm:class-II}
  The wheel-decorated honeycomb lattice graph is class~II.
\end{proposition}
\begin{proof}
  For the wheel-decorated honeycomb lattice graph $\mc G$, $\Delta(\mc G)=3$. At least four colors are required to properly edge color the broken wheel graph $W_5'$. The graph $W_5'$ is a subgraph of $\mc G$. Therefore, at least four colors are required to properly edge color $\mc G$.
\end{proof}

\subsection{Vertex coloring of lattice graphs}
The vertex coloring of a graph $G$ is an assignment of colors to the vertices of $G$. It is called proper if no two vertices of the same color share an edge. The minimum number required to properly color all vertices of a graph $G$ is known as the chromatic number $\chi(G)$. Edge coloring is intimately related to vertex coloring by the line graph $L(G)$ of $G$. The vertices of $L(G)$ are formed by the edges $e$ of $G$. There is an edge $\{e,e'\}$ in $L(G)$ if and only if $e$ and $e'$ share a vertex in $G$. An edge coloring of a graph $G$ is equivalent to the vertex coloring of the line graph $L(G)$ of $G$.

In this sense, \met{edge_coloring} gives a method for the minimal vertex coloring for any graph $H$ that is the line graph of some other graph $G$. If it exists, the inverse line graph $G=L^{-1}(H)$ can be found efficiently \cite{roussopoulos1973max}, so that, given an arbitrary graph $H=L(G)$, no prior knowledge of $G$ is required. Vice versa, any method for the vertex coloring of $L(G)$ gives a method for the edge coloring of $G$. The latter works for every $G$.

It is perhaps possible to alter \met{edge_coloring} and its analysis to obtain a method for the minimal or nearly minimal vertex coloring of general lattice graphs. By the above considerations, such an algorithm would be applicable for both the vertex coloring and the edge coloring of lattice graphs. Here, we briefly discuss such an alteration.

We remind that \met{edge_coloring} crucially relies on a simple and tight lower bound on the number of colors [\eq{vizing}]. In the case of the vertex coloring of general lattice graphs (not only those that are the line graph of another lattice graph), this role could be played by $\chi(G)\geq\max\{\om(G), \lceil\lvert V(G)\rvert / \alpha(G) \rceil\}$ \cite{lewis2016guide}, here $\om(G)$ is the number of vertices in a maximal clique of $G$, and $\alpha(G)$ is the number of vertices in a maximal independent set of $G$. In general, this bound is loose so that it fails to guarantee in general that a vertex coloring is minimal or nearly minimal even if the bound is saturated. However, in some cases there are close upper bounds. For example, according to the four color theorem \cite{appel1989every}, $\chi(G)\leq 4$ for planar graphs $G$ and by Brook's theorem \cite{brooks1941colouring}, $\chi\leq \Delta(G)$ unless $G$ is a complete graph or an odd cycle.

When the upper bounds coincide with the lower bounds, they completely determine $\chi(\mc G)$ for a lattice graph $\mc G$. In those cases, finding a vertex coloring using $\chi(\mc G)$ colors is guaranteed to be a \emph{minimal} vertex coloring. As examples, consider the asanoha lattice graph from \fig{laves}. It contains $K_4$, the complete graph on four vertices, as a subgraph, so that $\om(\mc G)\geq 4$. Since it is also planar, it follows that $\chi(\mc G)=4$. The star lattice graph from \fig{archimedean} contains $K_3$, the triangle graph, as a subgraph, so that $\chi(\mc G)\geq 3$. Since $\Delta(\mc G)=3$, $\chi(\mc G)=3$ by Brook's theorem.

In any case, for the alteration of \met{edge_coloring} to a vertex coloring method, more work is required to show how bounds on the chromatic number of finite (wrapped) patches imply bounds on the chromatic number of infinite lattice graphs. Additionally, more work would be required to translate the theorems and lemmas of the current section into the vertex coloring language, even when restricting the input graph to the vertex coloring method is the line graph of some other graph.

\section{Related work}\label{sec:previous_work_and_applications}
\subsection{Previous work}
The Bravais lattices were first considered by Frankenheim in 1835 \cite{frankenheim1835lehre} and later by Bravais in 1850 \cite{bravais1850memoire}, both within the context of crystallography. (A careful account of the origins of crystallography and Bravais lattices can be found in Ref.~\cite{burckhardt2013symmetrie}.) The concept of (geometric) lattice graphs is firmly established in solid-state physics; see, for example, the classic textbook by Ashcroft and Mermin from 1976 \cite{ashcroft1976solid}, or Hilbert and Cohn-Vossen's book from 1932 \cite{hilbert1932anschauliche} for an early account. Furthermore, there are software implementations for the explicit representation of (geometric) lattice graphs and the generation of finite patches \cite{bauer2011alps,todo2023lattice,dylan2023lattpy}. Nevertheless, we are unaware of a previous concise formal definition of (geometric) lattice graphs [as given in~\eq{lattice_graph}] in modern graph theoretic terms~\cite{harary1969graph}.

The edge coloring problem of finite graphs was first considered by Tait in 1880 in connection to the map coloring problem~\cite{tait1880remarks}. Since then, the edge coloring problem has found numerous applications, including the design of electrical networks and in scheduling problems~\cite{fiorini1977edge}. In 1916, König~\cite{konig1916graphen} showed that bipartite graphs are class~I (in modern terminology) and the proof gives a natural method for finding the minimal edge coloring of such graphs in polynomial time. Vizing's theorem [\eq{vizing}], was proven in 1964. Vizing also proved that if the maximum degree of a planar graph is between 2 and 5, the graph can be either class and conjectured that planar graphs with maximum degree at least 6 are class~I~\cite{vizing1965critical,vizing1968some}. Subsequently, it was established that planar graphs of degree at least 7 are class~I. (Vizing proved that planar graphs of max degree at least 8 are class~I~\cite{vizing1965critical}. Thirty-five years later, Zhang~\cite{zhang2000every} and Sanders and Zhao~\cite{sanders2001planar}
proved that planar graphs of degree 7 are class~I.) The $k$-uniform lattice graphs are generally not bipartite. Additionally, although they are planar, they have a maximum degree of at most 6. So, the the $k$-uniform lattice graphs are not restricted to class I a priori.

Relatively little attention has been paid to the edge coloring problem in the context of infinite graphs. A notable exception is the theorem of de Bruijn and Erd\H{o}s from 1951~\cite{bruijn1951colour}, which states that if every finite subgraph of an (infinite) graph $G$ is $k$-vertex-colorable, then $G$ is $k$-vertex-colorable. While this theorem concerns the vertex coloring problem, it extends to $k$-edge-colorings \cite{bosak1972chromatic, neumann1954embedding}. It is mentioned in a remark in Ref.~\cite{parviainen2005inclusions} that all Archimedean and Laves lattices are Vizing class~I, but no proof nor references are provided.

\subsection{Applications}

One important application of solving the edge coloring problem for lattice graphs lies in the quantum computation of crystals. Paradigmatic models for quantum magnetism include the Ising model and the Heisenberg (Anti)ferromagnetic model \cite{ashcroft1976solid}. In these models, spin 1/2 particles, such as electrons, are uniquely associated with the vertices $v$ of a given graph $G$ and a magnetic interaction is introduced between two particles if their associated vertices are adjacent in $G$. For periodic crystals, $G$ is a lattice graph.

As opposed to classical computers, for quantum computers there are known methods to simulate these systems accurately and efficiently in general \cite{lloyd1996universal,childs2021theory}. In quantum computation, a qubit $q$ is the quantum analogue of a classical bit. Physically, a single qubit is commonly embedded in a fixed physical entity, such as an (artificial) atom~\cite{leon2021materials,nielsen2010quantum}. A quantum computation is carried out by applying quantum gates to these qubits directly. Typically, each quantum gate acts on one or two qubits simultaneously. At any given time, a qubit can participate in at most one quantum gate. The sequence of gates needed to perform a quantum computation is called a quantum circuit.

For the quantum simulation of aforementioned crystals using a technique called Trotterization~\cite{lloyd1996universal,childs2021theory}, each vertex $v$ of the graph $G$ describing the crystal structure is also uniquely associated with a qubit $q_v$. A circuit is then constructed that consists of a repetition of a smaller quantum circuit, known as the cycle. The cycle contains a single type of two-qubit gate $U$ that implements the magnetic interaction between the particles associated with the qubits it acts on. To assure that every magnetic interaction is taken into account, the cycle is to be designed such that for each edge $\{v,v'\}$ in $G$, there is at least one gate between qubits $q_v$ and $q_{v'}$. Letting every edge in $G$ represent a gate of type $U$ and encoding the time step in which this gate is performed as the color of that edge, a minimal edge coloring of $G$ corresponds to a cycle for the quantum simulation of spin-1/2 particles on $G$ (interacting according to the Ising or Heisenberg Hamiltonian) using the fewest number of time steps in the cycle.

Optimal-depth quantum circuits are crucial for pushing the boundaries of current noisy intermediate-scale quantum technology, where noise effects in the output become stronger with the depth of the quantum circuit. For this reason, minimal and proper edge colorings of the kagome lattice graph were considered in Refs.~\cite{kattemolle2022variational,bosse2022probing, kattemolle2023line}, of the one-dimensional, square, and hexagonal lattice graphs in Ref.~\cite{burkard2022recipes}, and of the heavy-hex lattice graph in Ref. \cite{youngseok2023evidence}. Minimal edge colorings of the square, triangular, kagome and square-octagon lattice graphs were considered in~\cite{zhu2022efficient,chen2023efficient} for the related problem of the verification of the minimal energy quantum state of spin-1 Heisenberg model defined on those lattice graphs.

Closely tied to quantum simulation is quantum optimization through the quantum approximate optimization algorithm (QAOA) \cite{fahri2014quantum}. In QAOA, the optimization of the cost function defining the problem involves the quantum simulation of a (problem instance specific) Ising model. Thus, similar to the situation in quantum simulation, minimal edge colorings correspond to depth optimal circuits for quantum optimization \cite{majumdar2021optimizing,herrman2021lower}. For this reason, a minimal edge coloring of the square lattice was used in Ref.~\cite{harrigan2021quantum}. Additionally, a minimal edge coloring of the heavy-hex lattice (equivalent to the one in \fig{miscellaneous}) was used in Ref. \cite{weidenfeller2022scaling} in the compilation of quantum circuits for QAOA.

\section{Conclusion}
\subsection{Summary}
We formally defined lattice graphs, patches of lattice graphs and the wrapping of those patches. In \thm{proper_edge_coloring}, we proved that an edge coloring of a patch of a lattice graph induces a proper edge coloring of the entire lattice graph if and only if the wrapped patch is properly edge colored and self-loop-free. Furthermore, by \lem{max_degree}, the degree of this wrapped patch is equal to the degree of the lattice graph.

These results form the cornerstone of \met{edge_coloring} that, depending on what is requested, finds type-II or type-I colorings of lattice graphs. If a type-II coloring is requested, the method seeks a patch of the lattice graph that is self-loop free and permits a type-II coloring. Such a patch is always found (\lem{simple}, \app{proofs}). Subsequently, the wrapped patch is type-II colored, unwrapped (retaining the edge coloring) and returned. In \thm{correctness}, we showed that this colored patch induces a type-II coloring of the lattice graph.

If a type-I coloring is requested, the same steps are followed, with the exception that the wrapped patch must permit a type-I coloring instead of a type-II coloring (and that the wrapped patch is subsequently type-I colored instead of type-II colored). Unlike in the case of a type-II coloring, we now lack upper bounds on the required patch size. This is because the lattice graph itself may be class~II, thereby not permitting any type-I coloring. Additionally, even if the lattice graph is known or assumed to be class~I, upper bounds on the necessary patch size so that the wrapped patch is class~I are lacking. Nevertheless, when a self-loop-free and class~I wrapped patch is found, a type-I coloring of this patch induces a type-I coloring of the lattice graph (\thm{correctness}).

For the sake of allowing rigorous bounds, some flexibility of \met{edge_coloring} was removed, resulting in \alg{edge_coloring_algorithm}. This algorithm finds type-II (and hence also proper) colorings of lattice graphs. Its worst-case running time is $O(D^4\mu^2)$. Here, $\mu$ is the number of edges in the basis graph that induces the lattice graph, and $D=\max\{dx_{\max},dy_{\max}\}$ is the maximum cell coordinate entry of any vertex in the lattice graph's basis graph.

We implemented \met{edge_coloring} and obtained a type-I coloring of all Archimedean, Laves and $k$-uniform ($k\leq 6$) lattice graphs, thus showing that these are class~I. This demonstrates the practical applicability of the edge coloring method for obtaining type-I and hence \emph{minimal} edge colorings, despite the lack of bounds on the running time of the method in this case.  Nevertheless, in \pro{class-II}, we  constructed a planar lattice graph that is provably class~II, showing there are instances where \met{edge_coloring} does not halt. Still, if it is known a priori that a certain lattice graph is class~II, requesting a type-II coloring from our method produces a minimal edge coloring of that graph.

One important practical application of our work lies in quantum simulation, quantum optimization, and quantum state verification. When constructing (the cycles of) the quantum circuits for these applications, a two-qubit quantum gate must be added for each edge of a graph $G$ with qubits on its vertices. The graph $G$ is often a lattice graph. Since each qubit can participate in at most one quantum gate at a time, associating colors with layers of quantum gates, the depth-optimal cycle is described by a minimal edge coloring of $G$. As a result, our work provides an automated solution for finding depth-optimal quantum circuits for quantum simulation, state verification and quantum optimization.

\subsection{Outlook}
The practical applicability of our method was demonstrated by the minimal edge coloring of a plethora of lattice graphs. Thus, the most immediate practical extension of our work is the application of the edge coloring method to other lattice graphs, including those available in crystallographic databases \cite{grazulis2012crystallography,downs2003american}.

Our work also raises interesting theoretical questions. Firstly, considering that all $k$-uniform lattice graphs are class~I for $k\leq 6$, it is natural to hypothesize that \emph{all $k$-uniform lattice graphs are class~I}. These lattice graphs appear to offer sufficient structure to prove such a hypothesis; in a $k$-uniform geometric lattice graph, there are essentially only $k$ vertex different figures (the vertex plus adjacent edges). All other vertex figures in a geometric lattice graph are related to these by symmetries of the geometric lattice graph. However, a minimal edge coloring generally breaks the symmetries of a geometric lattice graph so that it is not clear how this structure could be exploited.

Secondly, it seems evident that, for a class~I lattice graph, as the patch size is increased, the wrapped patch must eventually permit a type-I coloring. Hence, it seems evident that our edge coloring method always terminates when a type-I coloring is requested of a class~I lattice graph. While it is true that a patch of any size of any class~I lattice graph is itself class~I (and hence permits a type-I coloring), when that patch is wrapped, it is no longer a subgraph of the lattice graph. Another possible obstruction in establishing an upper bound on the required patch size is that, for a planar graph, wrapping alters the genus of the surface in which the graph can be embedded.

Thirdly, our method extends to the edge coloring of tilings using symmetries beyond translational symmetry and in geometries beyond the Euclidean plane. In the terminology of Conway et al.~\cite{conway2008symmetries}, the orbifold of a tiling is obtained by merging or \emph{folding} all points related to one another by symmetries of the tiling. (One point and the points that can be generated from it by symmetries of the tiling form the \emph{orbit} of that point under the symmetry group of the tiling.) In this sense, the wrapping of a patch of a geometric lattice graph can essentially be seen as creating the orbifold of that lattice graph, only considering the translational symmetry of the lattice graph (modulo the patch size). This suggests that the wrapping stage of the edge coloring method can be altered to include wrapping by other symmetry subgroups. Furthermore, the wrapped patches need not be orbifolds (modulo some symmetry subgroups) of Euclidean tilings, opening the way to, e.g., the edge coloring of the meshes of tilings of the hyperbolic plane.

It may be possible to alter \met{edge_coloring} and its analysis to obtain analogous results for the minimal or nearly minimal \emph{vertex} coloring of lattice graphs. \met{edge_coloring} crucially relies on a lower bound on the number of colors that is both tight [\eq{vizing}] and can be obtained efficiently (\lem{max_degree}). Although numerous lower and upper bounds exist for the chromatic number $\chi(G)$ of finite graphs $G$, in general these bounds are not as applicable as Vizing's theorem in the case of edge coloring. Additionally, more work is required to show how bounds on the chromatic number of patches of a lattice graph relate to the chromatic number of infinite lattice graphs induced by those patches.

\vspace{1em}
\textbf{Data availability.}
The code and data are available as Supplemental Material~\cite{Note2}. It also includes images depicting minimal edge colorings of all 1318 lattice graphs discussed in this paper.

\begin{acknowledgements}
We thank G. Burkard, S. Storandt and R. de Wolf for useful discussion. We acknowledge funding from the Competence Center Quantum Computing Baden-W\"urttemberg, under the
project QORA II.
\end{acknowledgements}
\bibliography{bib.bib}

\appendix

\section{Proofs}\label{sec:proofs}
Here, we show the proofs for the theorems and lemmas in \sec{correctness} and \sec{running_time}, together with additional lemmas. The \emph{color class $c(C)$} of an edge colored graph $C$ is a set consisting of all edges with color $c$ \cite{fiorini1977edge}. (Not to be confused with Vizing's graph classes.) A matching of a multigraph is a set of edges containing no self-loops and where every vertex appears in at most one edge. If $C$ contains no self-loops and is properly edge colored, then the edge color class $c(C)$ forms a matching of $C$.

In the two lemmas below, let $B$ be a basis graph, with seed numbers $S$, that induces a lattice graph $\mc G$. By definition, $B$ is finite, simple, contains no redundant edges, and does not contain edges from nonseeds to nonseeds. Let $\tilde B=W(B)$ be the wrapped basis graph,\emph{assume $\tilde B$ is self-loop-free}, let $\tilde C$  be a proper edge coloring of $\tilde B$, and let $C=U(\tilde C)$ be the unwrapped $\tilde C$. Let $\mc C$ be the edge colored lattice graph induced by $C$. Because $C$ is an edge coloring of $B$, $\mc C$ is an edge coloring of $\mc G$. We have the following lemmas.

\begin{lemma}\label{thm:unique_seed}
  For every color $c$ of $C$, every seed number $s\in S$ appears in at most once in the color class $c(C)$.
\end{lemma}

\begin{proof}
  For all edges in $c(C)$, every seed number $s\in S$ can appear in at most one vertex of that edge because if, on the contrary, there is an edge $\{(dx,dy,s),(dx',dy',s)\}\in c(C)$, we have $\{(dx,dy,s),(dx',dy',s)\}\in E(B)$ which violates the assumption that $\tilde B$ is self-loop-free. To show that $s\in S$ can appear in at most one edge of $c(C)$, assume, on the contrary, that there exists edges $e_1,e_2\in c(C)$ with $(dx,dy,s)\in e_1$, $(dx',dy',s)\in e_2$ and $e_1\neq e_2$. Then, the vertex $(0,0,s)$ must have appeared in two edges of $c(\tilde C)$, in contradiction with the assumption that $\tilde C$ is a proper edge coloring.
\end{proof}

\begin{lemma}\label{thm:matching}
  For every color $c$ of $C$, the color class $c(\mc C)$ forms a matching of $\mc C$.
  \end{lemma}
\begin{proof}
  The lattice graph $\mc G$ is self-loop-free and hence $c(\mc C)$ is self-loop-free for all $c$. To show that, for all $c$, every vertex appears in at most one edge of $c(\mc C)$, assume, on the contrary, that there exist a color $c$ and a vertex that appears in two edges $e_1',e_2'$ of $c(\mc C)$. By definition of $\mc C$, these edges can be translated back to edges $e_1,e_2$ in $C$. Translation of edges conserves the color of those edges so that $e_1$ and $e_2$ are in $c(C)$. For $e_1'$ and $e_2'$ to share a vertex, it is necessary that the same seed number $s\in S$ appears in a vertex in $e_1'$ and $e_2'$. Since translation of an edge also conserves the seed number in each vertex of that edge, $s$ appears both in a vertex of $e_1\in c(C)$ and $e_2\in c(C)$, in contradiction with \lem{unique_seed}. Therefore, there is no vertex that appears in two edges of $c(\mc C)$ for all $c$.
\end{proof}

In the proof of \thm{proper_edge_coloring} we use a technique we call \emph{reseeding}. The patch $P_{n,m}(B)$, with $B$ a basis graph that induces $\mc G$, may be mapped by an isomorphism $I$ to a basis graph $B'$ that induces a lattice graph $\mc G'$ isomorphic to $\mc G$. Explicitly, this isomorphism $I$ may be constructed with the following steps. Assign a unique seed number $s'$ to every vertex $(x,y,s)\in V(P_{n,m})$ for which $0\leq x<n$ and $0\leq y<m$ and map $(x,y,s)$ to $(0,0,s')$. The vertices of the latter form are the seeds of $B'$. To obtain the nonseeds of $B'$, map the vertices $(x,y,s)\in V(P_{n,m})$ for which not $0\leq x<n$ or not $0\leq y<m$ to $(\lfloor x/n \rfloor,\lfloor y/m \rfloor,s')$, with $s'$ the new seed number of the seed $(x \mod n, y \mod m,s)$. Map all edges of $P_{n,m}(B)$ accordingly to obtain $B'$. The same isomorphism $I$ naturally extends to an isomorphism between an edge coloring $C$ of $P_{n,m}(B)$ and a reseeded edge coloring $C'$ of $B'$. Furthermore, the edge coloring $C'$ induces an edge coloring $\mc C'$ of $\mc G'$ that is isomorphic to the edge coloring $\mc C$ of $\mc G$ induced by~$C$.

In the following theorem and lemma, let the basis graph $B$ induce the lattice graph $\mc G$. Let $P_{n,m}$ be a patch of $\mc G$ of $n$ by $m$ basis graphs, $\tilde P_{n,m}=W(P_{n,m})$ the wrapped patch, and $\tilde C$ a proper edge coloring of the wrapped patch. Let $C=U(\tilde C)$ be the unwrapped $\tilde C$. Let $\mc C$ be the edge colored lattice graph induced by $C$. It forms an edge coloring of $\mc G$.

\rethm{proper_edge_coloring}{
The edge colored lattice graph $\mc C$ is properly edge colored if and only if the wrapped patch $\tilde P_{n,m}$ is self-loop-free.
}
\begin{proof}
  We first prove the theorem assuming that $n=m=1$ and then extend this proof to arbitrary $(n,m)$ by reseeding.

  Assume that $\tilde P_{1,1}$ is self-loop-free. Then by \lem{matching}, for all colors $c$ of the resulting edge colored patch $C$, $c(\mc C)$ forms a matching of $\mc C$. Since every edge of $\mc C$ is colored, the union of $c(C)$ for all $c$ equals $E(\mc G)$. Therefore, $\mc C$ is properly edge colored.

  Now assume that that $\tilde P_{1,1}$ contains self-loops. Because $P_{1,1}$ is self-loop-free, the self-loops in $\tilde P_{1,1}$ must be caused by wrapping. A self-loop in  $P_{1,1}$ occurs by wrapping if and only if $P_{1,1}$ contains an edge $\{v,v'\}$ with $v=(dx,dy,s),v'=(dx',dy',s)$. Note $s$ appears in $v$ and $v'$. The coloring basis graph $C$ will contain the edge $e=(\{v,v'\},c)$ for some $c$. Therefore, $\mc C$ contains the edge $e$ and the translated edge $e'=(\{v+(x,y,0),v'+(x,y,0)\},c)$ with $x=dx'-dx\in \mathbb Z,y=dy'-dy\in \mathbb Z$. Note $v+(x,y,0)=v'$ so that $e'=(\{v',v''\},c)$ for some $v''$. It follows $\mc C$ is not a proper edge coloring.

  If $n\neq 1$ or $m\neq 1$, let $B'$ be the basis graph obtained by reseeding $P_{n,m}$ using the isomorphism $I$. Let $\tilde B'$ be the wrapped $B'$ and let $C'$ be the properly edge colored and unwrapped $\tilde B'$. Let $\mc C'$ be the edge colored lattice graph induced by $C'$. By the arguments above, the edge colored lattice graph $\mc C'$ is properly edge colored if and only if $\tilde B'$ is self-loop-free. The edge colored lattice graph $\mc C'$ is properly edge colored if and only if $\mc C$ is properly edge colored. Furthermore, $\tilde B'$ is self-loop-free if and only if $\tilde P_{n,m}(B)$ is self-loop-free.
\end{proof}

The following lemma shows that at stage~3 (`color') of \met{edge_coloring}, we may determine the maximum degree of a lattice graph $\mc G$ by computing the maximum degree of the wrapped patch $\tilde P_{n,m}$. Definitions are as in \thm{proper_edge_coloring}.

\relem{max_degree}{
   If the wrapped patch $\tilde P_{n,m}$ is self-loop-free, then the maximum degree of the wrapped patch equals the maximum degree of the lattice graph;  $\Delta(\tilde P_{n,m})=\Delta(\mc G)$.
}
\begin{proof}
  By reseeding, we may assume $P_{n,m}$ is a basis graph $B$. We will show that, if a wrapped basis graph $\tilde B$ is self-loop-free, then for every vertex $v\in V(\tilde B)$ there is a $\bar v \in V(\mc G)$ of equal degree, and vice versa, from which the lemma follows.

  Take any $v\in V(\tilde B)$. Then, there is an $s\in \mathbb N_0$ such that $v=(0,0,s)$. We may assume without loss of generality that $i)$ $v\in V(B)$ and that $ii)$ all edges involving the seed number $s$ in $E(B)$ are of the form $\{(0,0,s),(dx',dy',s')\}$, with $dx',dy',s'\in \mathbb Z$ depending on the edge, and with $s\neq s'$ (for otherwise $\tilde B$ would have a self-loop).

  Assumption $i)$ is without loss of generality because if $v\notin V(B)$, then there must be a $(dx,dy)\neq (0,0)$ such that $(dx,dy,s)\in V(B)$. If $v\in V(\tilde B)$ is an isolated vertex, it is also isolated in $B$ and $\mc G$. The basis graph $B$ can then be altered by replacing $(dx,dy,s)\in V(B)$ by $(0,0,s)$. This alters $B$ but not $\tilde B$. Most importantly, the redefined $B$ is fully equivalent to the old $B$ because these basis graphs induce identical lattice graphs $\mc G$. If $v=(0,0,s)$ is not isolated in $\tilde B$, there is some edge in $\tilde B$ containing the vertex $v=(0,0,s)$. Consequently, $\{(0,0,s),(dx',dy',s')\}\in E(B)$ or $\{(dx,dy,s),(0,0,s')\}\in E(B)$ for some $dx,dy,dx',dy',s\neq s'$ . Then, $v\in V(B)$ follows from assumption $ii)$. Assumption $ii)$ is without loss of generality because any edge of the only possible alternative form $\{(dx,dy,s),(0,0,s')\}$ can be translated so that it becomes of the assumed form. This alters $B$ but not $\tilde B$ (except for the old edge labels, which is of no importance here) nor $\mc G$.

  Pick any $v=(0,0,s)\in V(\tilde B)$. Then also $v\in V(B)$ by $i)$. Unwrapping does not attach any edges to seeds such as $v$. Because of $ii)$, unwrapping also does not detach edges from $v$. So, $\deg_{\tilde B}(v)=\deg_{B}(v)$. Since $v\in V(B)$, $v\in V(\mc G)$. By construction of lattice graphs, $v$ cannot be incident to less edges in $\mc G$ than in $B$. The vertex $v$ can also not be incident to more edges in $\mc G$ than in $B$. To show the latter, note that any excess edge incident on $v$ in $\mc G$ must arise from a translated copy of $B$. However, by $ii)$ the edges in the translated copies of $B$ do not contain $v$. Thus, $\deg_B(v)=\deg_{\mc G}(v)$ and hence $\deg_{\tilde B}(v)=\deg_{\mc G}(v)$.

  Now, pick any $v' \in V(\mc G)$. By construction, the lattice graph $\mc G$ can be translated by an isomorphism $T$ so that $T(v)=(0,0,s)$. Note $(0,0,s)=:\bar v$ must be a vertex in $\tilde B$. Thus, we can reuse the previous result that $\deg_{\tilde B}(\bar v)=\deg_{\mc G}(\bar v)$.
\end{proof}

For the basis graphs $B$, which are finite, define
\begin{equation}
  \begin{aligned}\label{eq:max}
  dx_{\max}&=\max\{\lvert dx \rvert\mid(dx,dy,s)\in E(B)\},\\ dy_{\max}&=\max\{\lvert dy \rvert\mid(dx,dy,s)\in E(B)\}.
  \end{aligned}
\end{equation}
The following is used to show \met{edge_coloring} halts in case a proper edge coloring is requested ($t=3$).

\begin{lemma}\label{thm:self-loop-free}
  For any lattice graph $\mc G$ induced by a basis graph $B$, the wrapped patch $\tilde P_{n,m}(B)$, with $n\geq dx_{\mrm{max}}+1$, $m\geq dy_{\mrm{max}}+1$, is self-loop-free.
\end{lemma}
\begin{proof}
Consider the unwrapped patch $P_{n,m}$ with $n=dx_{\mrm{max}}+1$, $m=dy_{\mrm{max}}+1$ and the edge $\{v,v'\}\in P_{n,m}$. By construction of patches, we may assume that $v$ is a translated seed and that $v'$ is a translated seed or a translated nonseed. That is, we may assume there exist an $x\in \mathbb Z_n$ and a $y\in \mathbb Z_m$ such that $v=(0,0,s)+(x,y,0)$ and $v'=(dx,dy,s')+(x,y,0)$, with $(0,0,s),(dx,dy,s')\in V(B)$ and $dx,dy\in\mathbb Z$. Imposing periodic boundary conditions, $w$, acts trivially on $v$; $w(v)\equiv (x \mod n,y\mod m,s)=(x,y,s)$.

If $v'$ is also a translated seed, that is, if $dx,dy=0$, $w$ acts trivially on $v'$ as well and $\{w(v),w(v')\}$ cannot be a self-loop. Therefore, assume $dx\neq 0$ or $dy\neq 0$. If $dx\neq 0$, then $(dx+x) \mod n\neq x$. However, $(dx+x) \mod n= x$ is a necessary condition for $\{w(v),w(v')\})$ to be a self-loop. Similarly, if $dy\neq 0$,$(dy+y) \mod m\neq y$ although  $(dy+y) \mod m= y$ is a necessary condition for $\{w(v),w(v')\}$ to be a self-loop. Therefore, $\tilde P_{n,m}$ is self-loop-free.
\end{proof}

It is a stronger requirement that a wrapped patch becomes simple. The following lemma is used to show \met{edge_coloring} halts in case a type-II coloring is requested ($t=2$).

\begin{lemma}\label{thm:simple}
  For any lattice graph $\mc G$ induced by a basis graph $B$, the wrapped patch $\tilde P_{n,m}(B)$, with $n\geq 2dx_{\max}+1$ and $m\geq 2 dy_{\max}+1$, is simple.
\end{lemma}
\begin{proof}
  Because $2dx_{\max}\geq dx_{\max}$ and similarly for $dy_{\max}$, $\tilde P_{n,m}(B)$, with $n\geq 2dx_{\max}+1$ and $m\geq 2 dy_{\max}+1$, is self-loop-free by \lem{self-loop-free}. To show it does not contain multi-edges, assume, on the contrary, that a multi-edge occurs in $\tilde P_{n,m}(B)$ with $n\geq 2 dx_{\max}+1$ and $m\geq 2 dy_{\max}+1$. We will show that for a general multi-edge containing at least two (necessarily distinctly labeled) edges,  either those two edges are equal (i.e., also having the same label, a contradiction) or one edge is a redundant edge, which is ruled out by construction of $\mc G$.

  In general, for each edge in an unwrapped patch, imposing periodic boundary conditions $w$ acts nontrivially on at most one vertex. Thus, in case a nontrivial multi-edge occurs in $\tilde P_{n,m}(B)$, there exist two unequal edges $e_i=(\{v_i,w(v_i')\},\{v_i,v_i'\})\in\tilde E[P_{n,m}(B)]$ $(i=1,2)$, with
  \begin{equation}
  \begin{aligned}
    v_i&=(x_i,y_i,s_i)\\
    v_i'&=(x_i+dx_i,y_i+dy_i,s_i')\\
    w(v_i')&=\mathlarger ((x_i+dx_i) \mod n,(y_i+dy_i) \mod m, s'_i\mathlarger)
\end{aligned}
  \end{equation}
  For $\{e_1,e_2\}$ to be a multi-edge, we have $\{v_1,w(v_1')\}=\{v_2,w(v_2')\}$. This equality may occur in two distinct ways,
  \begin{equation}
    v_1=v_2       \quad \wedge \quad w(v_1')=w(v_2')\label{eq:case_I},
  \end{equation}
  or
  \begin{equation}
     v_1=w(v_2')   \quad \wedge \quad  w(v_1')=v_2. \label{eq:case_II}
  \end{equation}

  Consider the case of \eq{case_I}. Comparison of the first entries of the triples in \eq{case_I} gives $x_1=x_2$ and $(x_1+dx_1)\mod n =(x_2+dx_2)\mod n$. It follows that $dx_1\mod n = dx_2 \mod n$. Since $n\geq 2 dx_{\max}+1$, we have $dx_1=dx_2$. Using essentially the same arguments, it also follows that $y_1=y_2$ and $dy_1=dy_2$. Comparing the last entries in \eq{case_I}, we have $s_1=s_2$ and $s_1'=s_2'$. Therefore, $v_1=v_2$ and $v_1'=v_2'$, from which it follows that $e_1=e_2$, a contradiction.

  Now, consider the case of \eq{case_II}. Comparing the first entries of the triples in \eq{case_II}, we have $x_1=(x_2+dx_2)\mod n$ and $x_2=(x_1+dx_1)\mod n$, or, after using $0\leq x_i < n$ so that $x_i=x_i\mod n$, $0=(-x_1+x_2+dx_2)\mod n$ and $0=(-x_2+x_1+dx_1)\mod n$. Thus, there exist $\al_x,\al'_x\in\mathbb Z$ such that $-x_1+x_2+dx_2=\al_x n$ and $-x_2+x_1+dx_1=\al'_x n$. From comparing the second entries of the triplets in \eq{case_II}, it similarly follows that there exist $\al_y,\al'_y\in\mathbb Z$ such that $-y_1+y_2+dy_2=\al_y m$ and $-y_2+y_1+dy_1=\al'_y m$. Comparing the last entries in \eq{case_II}, $s_1=s_2'$ and $s_2=s_1'$. Therefore, in the case of \eq{case_II}, there exist $\al_x,\al_x',\al_y,\al_y'$ such that $\{v_1+(\al_xn,\al_ym,0),v_1'+(\al'_x n, \al'_y m ,0)\}=\{v_2,v_2'\}$, which means either $\{v_1,v_1'\}$ or $\{v_2,v_2'\}$ was redundant in the (not wrapped) patch $P_{n,m}(B)$. The latter occurs only if there is a redundant edge in $B$, which is ruled out by construction (\sec{lattice_graphs}).
\end{proof}

Let $B$ induce the lattice graph $\mc G$ and let $B,t$ be the input of \met{edge_coloring}.

\rethm{correctness}{
If $t=2$ or $t=3$, \met{edge_coloring} terminates and the output $C[P_{n,m}(B)],(n,m)$ induces a type-$t$ coloring of $\mc G$. If $t=1$ and \met{edge_coloring} terminates, the output induces a type-I coloring of $\mc G$.
}
\begin{proof}
  In stages~1--3, \met{edge_coloring} iterates through the sequence $a$ (defined in \met{edge_coloring}, stage~2), until a patch size $a_i$ is encountered such that $\tilde P_{a_i}$ is self-loop-free and type-$t$ colorable.

  Assume $t=2$ or $t=3$. For every $i$, either there is some $i'<i$ so that $\tilde P_{a_i'}$ is self-loop-free and type-$t$ colorable, or $\tilde P_{a_i}$ is encountered at stage~2. Take $i$ such that $a_i=(2dx_{\max}+1,2dy_{\max}+1)$. Then we move to stage~4  for some $i'<i$ or we reach $\tilde P_{(2dx_{\max}+1,2dy_{\max}+1)}$, with $dx_{\max}$ and $dy_{\max}$ as in \eq{max}. In case we reach $\tilde P_{(2dx_{\max}+1,2dy_{\max}+1)}$, by \lem{simple} and Vizing's theorem, $\tilde P_{(2dx_{\max}+1,2dy_{\max}+1)}$ is self-loop-free and type-$t$ colorable, so that we also move to stage~4 in this case.

  If instead $t=1$, there is no guarantee that an $(n,m)$ is encountered such that $\tilde P_{n,m}$ is self-loop-free and type-$t$ colorable. We may be stuck in stages~1--3 indefinitely.

  Nevertheless, for any $t\in\{1,2,3\}$, if $\tilde P_{n,m}$ is self-loop-free and type-$t$ colorable for some $(n,m)$, it is type-$t$ colored in stage~3, resulting in a graph $\tilde C$, which is passed to stage~4. It is guaranteed by \thm{proper_edge_coloring} that the unwrapped and edge colored patch $C$ induces a proper edge coloring $\mc C$ of $\mc G$. Additionally, by \lem{max_degree}, this proper edge coloring uses at most $\Delta(\mc G)$ $(t=1)$, $\Delta(\mc G)+1$ $(t=2)$ or any number of colors $(t=3)$. Therefore, $\mc C$ is a type-$t$ coloring of $\mc G$.
\end{proof}

\subsubsection{Running time}
Here, we derive the upper bound on the running time of \alg{edge_coloring_algorithm}, which is a specific use case of \met{edge_coloring}.

The algorithm receives a basis graph $B$ (that induces a lattice graph $\mc G$) as input. For simplicity of the analysis, we assume that $B$ does not have isolated vertices (which play no role in the edge coloring). The algorithm outputs a coloring basis graph $C$ [together with its size $(n,m)$] that induces a type-II edge coloring of $\mc G$. Notice that any type-II coloring is also proper edge coloring, so that the algorithm can also be used when a mere proper edge coloring is desired.

\realg{edge_coloring_algorithm}{
  Compute $dx_{\max},dy_{\max}$ from $B$ [\eq{max}]. Run \met{edge_coloring} with $t=2$ and the initial patch size $(n,m)$ set to $(2dx_{\max}+1,2dy_{\max}+1)$. Use the Misra and Gries edge coloring algorithm to edge color the finite patch in in stage~3.\hfill \qedsymbol
}

It is already shown in the proof of \thm{correctness} that this algorithm is correct and finite.

We now upper bound the worst-case running time of \alg{edge_coloring_algorithm} in the word RAM model of computation \cite{cormen2009algorithms}, where arithmetic operations on words (such as vertices) are counted as a single step. Using the adjacency-list representation of $B$, computing $dx_{\max},dy_{\max}$ takes $O(\lvert E(B) \rvert)$ steps. By \lem{simple} and the fact that the Misra and Gries edge coloring algorithm finds a type-II coloring of any finite simple graph, setting the initial patch size $(n,m)$ to $(2dx_{\max}+1,2dy_{\max}+1)$ assures that stages~1--3 are all run once, in sequence.

We now consider the running time of stages~1--3 of \met{edge_coloring}. Using that $B$ does not contain redundant edges, constructing the adjacency-list representation of $P_{n,m}$ takes $O(\lvert E(P_{n,m})\rvert)$ steps. Wrapping a single edge takes $O(1)$ steps, so that stage~1 of \met{edge_coloring} takes  $O(\lvert E(P_{n,m})\rvert)$ steps in total. Stage~2, checking for self-loops, also requires $O(\lvert E(\tilde P_{n,m})\rvert)=O(\lvert E(P_{n,m})\rvert)$ steps. (If fact, with the current initial $(n,m)$, this stage is redundant, but we keep it for consistency). The Misra and Gries edge coloring algorithm takes $O(\lvert E(\tilde P_{n,m}) \rvert \lvert V(\tilde P_{n,m}) \rvert)=O(\lvert E(P_{n,m}) \rvert^2)$ steps \cite{misra1992constructive,lewis2016guide}. Unwrapping a single edge takes $O(1)$ steps so that stage~4 of \met{edge_coloring} takes $O(\lvert E(P_{n,m})\rvert)$ steps.

So, the running time of \alg{edge_coloring_algorithm} is dominated by the running time of the Misra and Gries edge coloring algorithm, which is $O(\lvert E(P_{n,m}) \rvert^2)=O[(dx_{\max}dy_{\max}\lvert E(B) \rvert)^2]$ steps.
\end{document}